\documentclass[11pt]{article}

\usepackage{mathpazo}
\usepackage{amsfonts}
\usepackage{amsmath}
\usepackage{graphicx}
\usepackage{latexsym}
\usepackage{mathabx}
\usepackage{mathrsfs}

\usepackage[margin=1.05in]{geometry}
  
\usepackage[T1]{fontenc}
\usepackage{times}
\usepackage{color,graphicx}
\usepackage{array}
\usepackage{enumerate}
\usepackage{amsmath}
\usepackage{amssymb}
\usepackage{amsthm}
\usepackage{pgfplots}
\usepackage{pgf}
\usepackage{tikz}
\usetikzlibrary{patterns}
\usepgfplotslibrary{patchplots} 
\usetikzlibrary{pgfplots.patchplots} 
\pgfplotsset{width=9cm,compat=1.5.1}

\usepackage[english]{babel}
\usepackage{amsthm}
\newtheorem{theorem}{Theorem}

\newtheorem{definition}[theorem]{Definition}
\newtheorem{lemma}[theorem]{Lemma}
\newtheorem{corollary}[theorem]{Corollary}
\newtheorem{remark}[theorem]{Remark}
\newtheorem{proposition}[theorem]{Proposition}

\newtheorem{example}[theorem]{Example}

\usepackage{xcolor}
\usepackage{makeidx}
\usepackage[colorlinks=true,linkcolor=blue,anchorcolor=blue,citecolor=red,urlcolor=magenta]{hyperref}
\usepackage[alphabetic,backrefs]{amsrefs}

\begin{document}
\title{Sedentariness in quantum walks}

\author{
	Hermie Monterde,\textsuperscript{\!\!1}
}

\maketitle










\begin{abstract}

We formalize the notion of a sedentary vertex and present a relaxation of the concept of a sedentary family of graphs introduced by Godsil [\textit{Linear Algebra Appl}. 614:356-375, 2021]. We provide sufficient conditions for a given vertex in a graph to exhibit sedentariness. We also show that a vertex with at least two twins (vertices that share the same neighbours) is sedentary. We prove that there are infinitely many graphs containing strongly cospectral vertices that are sedentary, which reveals that, even though strong cospectrality is a necessary condition for pretty good state transfer, there are strongly cospectral vertices which resist high probability state transfer to other vertices. Moreover, we derive results about sedentariness in products of graphs which allow us to construct new sedentary families, such as Cartesian powers of complete graphs and stars.

\end{abstract}

	


\addtocounter{footnote}{1}
\footnotetext{Department of Mathematics, University of Manitoba, Winnipeg, MB, Canada R3T 2N2}

\tableofcontents

\section{Introduction}\label{secINTRO}

A \textit{(continuous-time) quantum walk} (CTQW) on $X$ describes the propagation of quantum states across a quantum spin network modelled by the graph $X$, where the qubits in the spin network and the interactions between them are represented by the vertices and edges of $X$, respectively. If $H$ is a real symmetric matrix that encodes the adjacencies in $X$, i.e., $H_{u,v}=0$ if and only if $[u,v]$ is not an edge in $X$, then the CTQW on $X$ with respect to $H$ is determined by the complex symmetric unitary matrix
\begin{equation}
\label{tM}
U_H(t)=e^{itH}.
\end{equation}
We call $U_H(t)$ and $H$ resp.\ the \textit{transition matrix} and the \textit{Hamiltonian} of the quantum walk. If $H$ in (\ref{tM}) is clear from the context, then we write $U_H(t)$ as $U(t)$. As $U(t)$ is unitary, $\sum_{j\in V(X)}\lvert U(t)_{u,j}\rvert^2=1$ for all $t$, and so $\lvert U(t)_{u,v}\rvert^2$ is interpreted as the probability of quantum state transfer between $u$ and $v$ at time $t$.

The concept of a sedentary family of graphs was introduced by Godsil \cite{SedQW}, which was mainly motivated by the behaviour of quantum walks on complete graphs. Godsil defined a family $\mathscr{F}$ of graphs to be \textit{sedentary} if there is a constant $a>0$ such that for each $X\in \mathscr{F}$ and each vertex $u$ of $X$, we have $\lvert U(t)_{u,u}\rvert\geq 1-\frac{a}{\lvert V(X)\rvert}$ for all $t$. Consequently, $\lvert U(t)_{u,u}\rvert$ tends to 1 for each $u\in V(X)$ as $\lvert V(X)\rvert$ increases. Godsil showed that large classes of strongly regular graphs are sedentary at any vertex. While the main focus in \cite{SedQW} was to study sedentary families of graphs, Godsil also investigated sedentariness of a single vertex in a graph by showing that cones over $k$-regular graphs exhibit varying degrees of sedentariness at their apexes with respect to the adjacency matrix depending on the value $k$. Frigerio and Paris showed that cones are also sedentary at the apex with respect to the Laplacian matrix \cite{chiral}. To the best of our knowledge, no other families of finite graphs are known to exhibit sedentariness.

In order to better understand the notion of single vertex sedentariness and obtain more sedentary families of graphs, we formalize the definition of a sedentary vertex and propose a relaxation of the notion of a sedentary family of graphs. Let $0<C\leq 1$ be a constant. We say that a vertex $u$ of $X$ is \textit{$C$-sedentary} if $\lvert U(t)_{u,u}\rvert\geq C$ for all $t$. We say that a family $\mathscr{F}$ of graphs is \textit{$C$-sedentary} if there exists a real-valued function $f$ satisfying $0< f(s)\leq 1$ for all $s>0$ such that (i) for each $X\in \mathscr{F}$, some vertex $u$ of $X$ is $f(\lvert V(X)\rvert)$-sedentary and (ii) $f(s)\rightarrow C$ as $s$ increases. If $f(s)=1-\frac{a}{s}$ for some $a>0$, then $C=1$, and if we add that each vertex of each $X\in \mathscr{F}$ is $f(\lvert V(X)\rvert)$-sedentary, then the concept of a $C$-sedentary family coincides with Godsil's notion of a sedentary family. If $C=0$ and $\inf_{t>0}\lvert U(t)_{u,u}\rvert=f(\lvert V(X)\rvert)$ for each $X\in\mathscr{F}$, then the family is said to be \textit{quasi-sedentary}, a concept first introduced in this paper. We emphasize that these properties depend on the matrix $H$, which we later choose to be a generalized adjacency matrix or a generalized normalized adjacency matrix of $X$.

The main goal of this paper is to provide sufficient conditions for $C$-sedentariness of a vertex and construct families of graphs that are $C$-sedentary. We prove our main result, which states that by an appropriate choice of a subset $S$ of the eigenvalue support of a vertex $u$, one may be able to show that $u$ is sedentary. We then use this result to establish that for any vertex $u$ in a set of twins $T$, $\lvert U(t)_{u,u}\rvert\geq 1-\frac{2}{\lvert T\rvert}$ for all $t$. Consequently, vertices with at least two twins are sedentary, which allows us to construct new families of graphs that are $C$-sedentary. This includes graphs built from joins, graphs with tails and blow-ups of graphs. We also show that sedentariness is preserved under Cartesian products, which provides another way to construct $C$-sedentary families. Another result, which is rather unexpected, is that there are infinitely many graphs containing strongly cospectral vertices that are sedentary. This reveals that some strongly cospectral vertices resist high probability transfer to other vertices. We also discuss the connection of sedentariness to other types of quantum state transfer. Even though sedentary vertices do not exhibit pretty good state transfer, we show that there are $C$-sedentary and quasi-sedentary families whose each member graph exhibits proper fractional revival at the sedentary vertices. For local uniform mixing, we show that this is only possible for quasi-sedentary families.

Throughout this paper, we assume that $X$ is a connected weighted undirected graph with possible loops but no multiple edges. We denote the vertex and edge sets of $X$ resp.\ by $V(X)$ and $E(X)$, and we allow the edges of $X$ to have nonzero real weights (i.e., an edge can have either positive or negative weight). We denote an edge between vertices $u$ and $v$ by $[u,v]$. We say that $X$ is \textit{simple} if $X$ has no loops, and $X$ is \textit{unweighted} if all edges of $X$ have weight one. For $u\in V(X)$, we denote the set of neighbours of $u$ in $X$ as $N_X(u)$, and the characteristic vector of $u$ as $\textbf{e}_u$, which is a vector with a $1$ on the entry indexed by $u$ and $0$'s elsewhere. The all-ones vector of order $n$, the zero vector of order $n$, the $m\times n$ all-ones matrix, and the $n\times n$ identity matrix are denoted resp. by $\textbf{1}_n$, $\textbf{0}_n$, $\textbf{J}_{m,n}$ and $I_n$. If $m=n$, then we write $\textbf{J}_{m,n}$ as $\textbf{J}_n$, and if the context is clear, then we simply write these matrices resp. as $\textbf{1}$, $\textbf{0}$, $\textbf{J}$ and $I$. We also represent the transpose of $M$ by $M^T$. We denote the simple unweighted empty, cycle, complete, and path graphs on $n$ vertices resp. as $O_n$, $C_n$, $K_n$, and $P_n$. We also denote the simple unweighted complete bipartite graph with partite sets of sizes $n_1,\ldots,n_k$ as $K_{n_1,\ldots,n_k}$.

For two graphs $X$ and $Y$, the \textit{join} $X\vee Y$ is the resulting graph after adding all edges $[u,v]$ of weight one, where $u\in V(X)$ and $v\in V(Y)$, while the \textit{union} $X\cup Y$ is the resulting graph with $V(X\cup Y)=V(X)\cup V(X)$ and $E(X\cup Y)=E(X)\cup E(Y)$. The \textit{Cartesian product} $X\square Y$ is a graph with vertex set $V(X)\times V(Y)$ where $(u,x)$ and $(v,y)$ are adjacent in $X\square Y$ if either $u=v$ and $[x,y]$ is an edge in $Y$ or $x=y$ and $[u,v]$ is an edge in $X$. The weight of the edge between $(u,x)$ and $(v,y)$ is equal to the weight of $[u,v]$ if $x=y$ and $[x,y]$ if $u=v$. The \textit{direct product} $X\times Y$ is the graph with vertex set $V(X)\times V(Y)$ where $(u,x)$ and $(v,y)$ are adjacent in $X\times Y$ if $[u,v]$ and $[x,y]$ are edges resp. in $X$ and $Y$. The weight of the edge between $(u,x)$ and $(v,y)$ is equal to the product of the weights of the edges $[u,v]$ and $[x,y]$.

We define the adjacency matrix $A(X)$ of $X$ entrywise as
\begin{equation}
A(X)_{u,v}=
\begin{cases}
 \omega_{u,v}, &\text{if $u$ and $v$ are adjacent}\\
 0, &\text{otherwise},
\end{cases}
\end{equation}
where $0\neq \omega_{u,v}\in\mathbb{R}$ is the weight of $[u,v]$. The degree matrix $D(X)$ of $X$ is the diagonal matrix of vertex degrees of $X$, where $\operatorname{deg}(u)=2\omega_{u,u}+\sum_{j\neq u}\omega_{u,j}$ for each $u\in V(X)$. We say that $X$ is weighted-regular if $D(X)$ is a scalar multiple of the identity matrix. As $X$ is weighted, it is possible that $\operatorname{deg}(u)=0$ without $u$ being isolated. Assuming that $\operatorname{deg}(u)\geq 0$ for all $u\in V(X)$, we define $D(X)^{-\frac{1}{2}}$ as the diagonal matrix whose $(u,u)$ entry is $1/\sqrt{\operatorname{deg}(u)}$ if $\operatorname{deg}(u)>0$ and $0$ otherwise.

Let $a,b,c\in\mathbb{R}$ with $a\neq 0$. A matrix of the form $\mathbf{A}(X)= c I+b D(X)+a A(X)$ is called a \textit{generalized adjacency matrix} $\mathbf{A}(X)$ of $X$, and a matrix of the form $ \mathcal{A}(X)=b I+a D(X)^{-\frac{1}{2}}A(X)D(X)^{-\frac{1}{2}}$ is called a \textit{generalized normalized adjacency matrix} $\mathcal{A}(X)$ of $X$. These two matrices were first studied in \cite{MonterdeELA} in the context of quantum state transfer (in particular, in relation to the concept of strong cospectrality). We consider these two types of matrices, which are generalizations of well-known matrices associated to graphs. Indeed, if $c=0$, $b=1$ and $a=- 1$, then $\mathbf{A}$ becomes the \textit{Laplacian matrix} $L(X)$ of $X$, while $\mathcal{A}(X)$ becomes the \textit{normalized Laplacian} matrix $\mathcal{L}(X)$ of $X$. Since the quantum walks determined by $b I+aH$ and $H$ are equivalent, we simplify the discussion by considering the matrices
\begin{equation}
\label{gen}
\mathbf{A}(X)=\alpha D(X)+A(X)\quad \text{and}\quad \mathcal{A}(X)=D(X)^{-\frac{1}{2}}A(X)D(X)^{-\frac{1}{2}},
\end{equation}
and note that the quantum walks determined by $\textbf{A}$ are equivalent for all $\alpha\in\mathbb{R}$ whenever the graph is weighted-regular. We use $M(X)$ to denote $\textbf{A}(X)$ or $\mathcal{A}(X)$, and use $H=M(X)$ in (\ref{tM}). If the context is clear, then we write $M(X)$, $A(X)$, $L(X)$, $\mathcal{L}(X)$ and $D(X)$ resp. as $M$, $A$, $L$, $\mathcal{L}$ and $D$. Finally, if $U_{X\square Y}(t)$ is the transition matrix of $X\square Y$ with respect to $\mathbf{A}$, then it is known that
\begin{equation}
\label{Ucart}
\begin{split}
U_{X\square Y}(t)=U_X(t)\otimes U_{Y}(t),
\end{split}
\end{equation}
while if $U_{X\times Y}(t)$ is the transition matrix of $X\times Y$ with respect to $\mathcal{A}$, then $U_{X\times Y}(t)=U_X(t)\otimes U_{Y}(t)$ whenever $X$ and $Y$ are simple. Here, $A\otimes B$ denotes the Kronecker product of matrices $A$ and $B$.

\section{Sedentariness}\label{secSed}

We begin with the definition of a sedentary vertex.

\begin{definition}
\label{def}
We say that vertex $u$ of $X$ is \textit{$C$-sedentary} if for some constant $0<C\leq 1$,
\begin{equation}
\label{inf}
\inf_{t>0}\lvert U_M(t)_{u,u}\rvert\geq C.
\end{equation}
If equality holds in (\ref{inf}), then we say that $u$ is \textit{sharply} $C$-sedentary, while if the infimum in (\ref{inf}) is attained for some $t>0$, then we say that $u$ is \textit{tightly} $C$-sedentary.
\end{definition}
We also say that $u$ is \textit{not sedentary} if $\inf_{t>0}\lvert U_M(t)_{u,u}\rvert=0$. Note that for a sharply $C$-sedentary vertex, $C$ is the best lower bound one can get for $\lvert U_M(t)_{u,u}\rvert$ for all $t$. It is also clear that a tightly sedentary vertex is sharply sedentary, but the converse is not true. If $C$ is not important, then we resp. say sedentary, sharply sedentary, and tightly sedentary. Sedentariness of $X$ at $u$ implies that $\lvert U_M(t)_{u,u}\rvert$ is bounded away from $0$, and as a result, the quantum state initially at vertex $u$ tends to stay at $u$.

As $M$ is real symmetric, it admits a spectral decomposition $M=\sum_{j}\lambda_jE_j$, and so we can write (\ref{tM}) as
\begin{equation*}
U_M(t)=\sum_{j}e^{it\lambda_j}E_j,
\end{equation*}
where the $\lambda_j$'s are the distinct eigenvalues of $M$ and $E_j$ is the orthogonal projection matrix associated with $\lambda_j$. The \textit{eigenvalue support} of vertex $u$ with respect to $M$ is the set
$\sigma_u(M)=\{\lambda_j:E_j\textbf{e}_u\neq \textbf{0}\}$. We say that two vertices $u$ and $v$ are \textit{cospectral} if $(E_j)_{u,u}=(E_j)_{v,v}$ for each $j$. It is immediate that if $X$ has an automorphism mapping $u$ to $v$, then they are cospectral.

Let $S_1$ and $S_2$ be two non-empty disjoint proper subsets of $V(X)$. Order the vertices of $X$ in a way that $S_1$ comes first, followed by $S_2$ and then $V(X)\backslash (S_1\cup S_2)$. We say that there is \textit{pretty good state transfer} (PGST) from $S_1$ and $S_2$ if for each $\epsilon>0$, there exists a time $t_{\epsilon}$ such that $U_M(t_{\epsilon})$ has the block form
\begin{equation*}
U_M(t_{\epsilon})=\left[\begin{array}{cccc} *&U_{\epsilon}&*\\ U^T_{\epsilon}&*&* \\ *&*&*\end{array} \right],
\end{equation*}
where $U_{\epsilon}$ is an $\lvert S_1\rvert$-by-$\lvert S_2\rvert$ matrix satisfying $\|U_{\epsilon}\|>1-\epsilon$. Clearly, if $\lvert S_1\rvert=\lvert S_2\rvert=1$, then we get PGST between two vertices. If $\|U_{\epsilon}\|=1$, then we say that \textit{perfect state transfer} (PST) occurs from $S_1$ and $S_2$, a notion that is equivalent to \textit{group state transfer} (GST) introduced by Brown et al.\ \cite{Brown2021}. As $U_M(t_1)$ is non-singular, if PST occurs from $S_1$ and $S_2$, then $\lvert S_1\rvert\leq \lvert S_2\rvert$, and the case $\lvert S_1\rvert=\lvert S_2\rvert=1$ yields PST between two vertices. PST, PGST and GST (and later on, uniform mixing and fractional revival) fall under the general notion of quantum state transfer, which is an important physical concept.

The following basic properties of $C$-sedentary vertices are immediate from the fact that $U(t)$ is unitary.

\begin{proposition}
\label{prop}
Let $X$ be a graph with vertex $u$.
\begin{enumerate}
\item If $X$ is (sharply or tightly) $C_1$-sedentary at $u$, where $0< C_1\leq 1$, then the following hold.
\begin{enumerate}
\item $X$ is also $C_2$-sedentary at $u$ whenever $0<C_2\leq C_1$.
\item If $u$ and $v$ are cospectral vertices, then $X$ is also (sharply or tightly) $C_1$-sedentary at $v$.
\item Any subset $S$ of $V(X)$ containing $u$ cannot be involved in pretty good state transfer in $X$.
\item For any vertex $v\neq u$, $\sup_{t>0}\lvert U_M(t)_{u,v}\rvert\leq \sqrt{1-C_1^2}$.
\end{enumerate}
\item If for each vertex $v\neq u$ of $X$, there is a constant $C_v<1$ such that $\sup_{t>0}\lvert U_M(t)_{u,v}\rvert=C_v$, then $u$ is sharply $C$-sedentary if and only if $1-\sum_{v\neq u}C_v^2>0$, in which case $C=1-\sqrt{\sum_{v\neq u}C_v^2}$.
\end{enumerate}
\end{proposition}

By Proposition \ref{prop}(1a), it is desirable to find the least $C<1$ such that a vertex is $C$-sedentary, i.e., the $C$ such that $u$ is sharply $C$-sedentary. Proposition \ref{prop}(1c) implies that PGST and sedentariness are mutually exclusive. Thus, our investigation of sedentariness is motivated in the same way as the study of PGST, in a sense that identifying sedentary vertices rules out the existence of PGST. Proposition \ref{prop}(1d) tells us that a necessary condition for sedentariness of $u$ is that $\lvert U(t)_{u,v}\rvert$ is bounded away from 1 for any $v\neq u$, while Proposition \ref{prop}(2) provides a sufficient condition for sedentariness. But since not much is known about pairs of vertices such that $\lvert U(t)_{u,v}\rvert$ bounded away from 1, Proposition \ref{prop}(2) will not be very useful to us. Instead, we present a sufficient condition for sedentariness in Section \ref{secSuff} that only depends on the diagonal entries of the $E_j$'s. Next, we define what it means for a family of graphs to be $C$-sedentary.

\begin{definition}
\label{defs}
Let $0\leq C\leq 1$ and $\mathscr{F}$ be a (countable) family of graphs. We call $\mathscr{F}$ is \textit{$C$-sedentary} if there is a function $f:\mathbb{R}^+\rightarrow (0,1]$ such that (i) for each $X\in\mathscr{F}$ and some $u\in V(X)$, the graph $X$ is $f(\lvert V(X)\rvert)$-sedentary at $u$ and (ii) $f(s)\rightarrow C$ as $s\rightarrow\infty$. Further, if $C=1$, then we call $\mathscr{F}$ is \textit{sedentary}; if each $X\in\mathscr{F}$ is sharply (resp., tightly) $f(\lvert V(X)\rvert)$-sedentary at $u$, then call $\mathscr{F}$ is \textit{sharply} (resp., tightly) $C$-sedentary; and if $C=0$ and $\mathscr{F}$ is sharply $C$-sedentary, then call $\mathscr{F}$ is \textit{quasi-sedentary}.
\end{definition}

Note that if $C=1$, then the above notion coincides with Godsil's definition of sedentary quantum walks. For example, if $\mathscr{K}$ is the family of complete graphs on $n\geq 3$ vertices, then for all $t$,
\begin{equation}
\label{K}
\lvert U(t)_{u,u}\rvert=\frac{\lvert n-1+e^{itn}\rvert}{n}\geq 1-\frac{2}{n}
\end{equation}
with equality if and only if $t=\frac{j\pi}{n}$ for odd $j$. Thus, $\mathscr{K}$ is a sedentary family. The case $0\leq C<1$ is a more relaxed version of $C$-sedentariness than the case $C=1$. In \cite{SedQW}, Godsil showed that cones on $d$-regular graphs on $n$ vertices are $C$-sedentary at the apex with respect to $A$, where $C=\frac{d^2}{d^2+4n}$. Thus, the concept of $C$-sedentariness for $0<C\leq 1$ is not entirely new, although this concept is first formalized in this paper. Quasi-sedentariness, on the other hand, is a new concept introduced in this paper, and may be regarded as the weakest form of sedentariness for families of graphs.

\section{Products}\label{secProd}

Consider $\textbf{A}$ and $\mathcal{A}$ in (\ref{gen}). In this section, we derive results about sedentariness in products of graphs.

\begin{theorem}
\label{Cart1}
Let $X_1,\ldots,X_n$ be weighted graphs with possible loops, and $Z=\square_{j=1}^n X_j$.
\begin{enumerate}
\item If each $X_j$ is $C_j$-sedentary at $u_j$, then $Z$ is $\prod_{j=1}^n C_j$-sedentary at $(u_1,\ldots,u_n)$. In particular, if each $X_j$ is sharply $C_j$-sedentary at $u_j$, then $Z$ is sharply $C'$-sedentary at $(u_1,\ldots,u_n)$ with $C'\geq \prod_{j=1}^n C_j$.
\item If $Z$ is $C$-sedentary at $(u_1,\ldots,u_n)$, then each $X_j$ is sharply $C_j$-sedentary at $u_j$ for some $0< C_j\leq 1$.
\item If each $X_j$ is tightly $C_j$-sedentary at $u_j$ and there exists a time $t_1$ such that $\lvert U_{X_j}(t_1)_{u_j,u_j}\rvert=C_j$ for each $j$, then $Z$ is tightly $C$-sedentary at $(u_1,\ldots,u_n)$ with $\lvert U_{Z}(t_1)_{(u_1,\ldots,u_n),(u_1,\ldots,u_n)}\rvert=\prod_{j=1}^n C_j$.
\end{enumerate}
\end{theorem}
\begin{proof}
From (\ref{Ucart}), we have $\lvert U_{X_1\square X_2}(t)_{(u_1,u_2),(u_1,u_2)}\rvert=\lvert U_{X_1}(t)_{u_1,u_1}\rvert\cdot \lvert U_{X_2}(t)_{u_2,u_2}\rvert$. Using the fact that $\inf _{t>0}f(t)g(t)\geq \inf_{t>0} f(t)\inf_{t>0} g(t)$ for all nonnegative functions $f$ and $g$ yields (1-3).
\end{proof}

By Theorem~\ref{Cart1}, Cartesian products of graphs with sedentary vertices also contain sedentary vertices. Consequently, Cartesian products of sedentary families also yield a sedentary family.

\begin{corollary}
\label{Cart3}
Let $\mathscr{F}_1,\ldots,\mathscr{F}_n$ be families of weighted graphs with possible loops. If each $\mathscr{F}_j$ is $C_j$-sedentary, then $\mathscr{F}=\left\{\square_{j=1}^n X_j:X_j\in \mathscr{F}_j\right\}$ is $\prod_{j=1}^n C_j$-sedentary.
\end{corollary} 

If the graphs involved are simple, then Theorem~\ref{Cart1} and Corollary~\ref{Cart3} also hold for the direct product. We also note that if $X$ and $Y$ are simple and weighted-regular, then so are $X\square Y$ and $X\times Y$, and so the quantum walks determined by $\textbf{A}$ and $\mathcal{A}$ are equivalent. In this case, Theorem~\ref{Cart1} and Corollary~\ref{Cart3} apply to $\mathcal{A}$, and their analogs for the direct product also apply to $\mathbf{A}$. However, if $X$ is not weighted-regular, then it is not clear how to obtain simple expressions for $e^{it\mathcal{A}(X\square Y)}$ and $e^{it\mathbf{A}(X\times Y)}$.

Next, we examine Cartesian products of complete graphs. Since these are regular, our results apply to $\textbf{A}$ and $\mathcal{A}$. We use $\nu_2(b)$ to denote the largest power of two that divides an integer $b$.

\begin{theorem}
\label{rook}
Let $n_1,\ldots,n_m\geq 2$ and $X=\square_{j=1}^mK_{n_j}$. The following hold.
\begin{enumerate}
\item If $n_j=2$ for some $j$, then $X$ is not sedentary at any vertex.
\item If each $n_j\geq 3$, then $X$ is $C$-sedentary at any vertex, where $C=\prod_{j=1}^m(1-\frac{2}{n_j})$. In particular, if the $\nu_2(n_j)$'s are all equal, then $\lvert U_{X}(t_1)_{w,w}\rvert\geq C$ for any vertex $w$ with equality at  $t_1=\pi/2^{\nu_2(n_1)}$.
\end{enumerate}
\end{theorem}

\begin{proof}
Let $n_j=2$. If $X$ is sedentary at some vertex $(u_1,\ldots,u_m)$, where $u_j\in V(K_2)$, then Theorem~\ref{Cart1}(2) implies that $K_2$ is sedentary at $u_j$, which is a contradiction because $K_2$ exhibits PST. This proves (1). Now, if each $n_j\geq 3$, then (\ref{K}) and Theorem~\ref{Cart1} imply that  $X$ is $C$-sedentary. If we add that the $\nu_2(n_j)$'s are all equal, then each $n_j'=n_j/2^{\nu_2(n_1)}$ is odd, and so (\ref{K}) implies that each $K_{n_j}$ is tightly $(1-\frac{2}{n_j})$-sedentary at any vertex at time $t_1=\pi/2^{\nu_2(n_1)}$. Invoking Theorem~\ref{Cart1}(3) completes the proof of (2).
\end{proof}

The following corollary is immediate from Theorem~\ref{rook}.

\begin{corollary}
\label{rook1}
Fix $k$ and let $\mathscr{F}$ be a family of graphs of the form $\square_{j=1}^kK_{n_j}$, where each $n_j\geq 3$.
\begin{enumerate}
\item If $n_j$ is fixed for some $j$, then $\mathscr{F}$ is $(1-\frac{2}{n_j})$-sedentary at any vertex. 
\item If each $n_j$ increases as $\prod_{j=1}^k n_j\rightarrow\infty$, then $\mathscr{F}$ is sedentary at any vertex.
\end{enumerate}
If we add that the $\nu_2(n_j)$'s are equal for all $\square_{j=1}^kK_{n_j}\in \mathscr{F}$, then the sedentariness in (1) and (2) is tight.
\end{corollary}

The Hamming graph $H(k,n)$ is obtained by taking the Cartesian product of $k\geq 1$ copies of $K_n$. Combining Theorem~\ref{rook}(2) and Corollary~\ref{rook1}(2) yields the following result about Hamming graphs.

\begin{corollary}
\label{ham}
Let $u$ be a vertex of $H(k,n)$ and $\mathscr{F}$ be a family of Hamming graphs $H(k,n)$. If $n\geq 3$, then $u$ is tightly sedentary in $H(k,n)$ and $\mathscr{F}$ is a sedentary family of graphs.
\end{corollary}

\section{A sufficient condition}\label{secSuff}

We say that vertex $u$ is \textit{periodic} in $X$ with respect to $M$ if $\lvert U_M(t_1)_{u,u}\rvert=1$ for some time $t_1$, and the minimum such $t_1>0$ is called the minimum period of $u$, denote by $\rho$. If $u$ is periodic, then $\lvert U_M(t)_{u,u}\rvert$ is a periodic function because $\lvert U_M(t+\rho)_{u,u}\rvert=\lvert (U_M(t)U(\rho))_{u,u}\rvert=\lvert U_M(t)_{u,u}\rvert\cdot \lvert U_M(\rho)_{u,u}\rvert=\lvert U_M(t)_{u,u}\rvert$. In this case, $\inf_{t>0}\lvert U_M(t)_{u,u}\rvert=\min_{t\in [0,\rho]}\lvert U_M(t)_{u,u}\rvert$, and so the following is immediate.

\begin{lemma}
\label{tight}
If $u$ is periodic, then $u$ is tightly sedentary if and only if $U_M(t)_{u,u}\neq 0$ for all $t\in [0,\rho]$.
\end{lemma}

From Lemma~\ref{tight}, a periodic sedentary vertex is tightly sedentary. Since a rook graph has all integer eigenvalues, it is periodic. By Theorem~\ref{rook}(2), it follows that each vertex in a rook graph is tightly sedentary.

\begin{example}
By Theorem~\ref{rook}(2), the rook graphs $X=K_3\square K_4$ and $Y=K_3\square K_5$ resp.\ are $\frac{1}{6}$- and $\frac{1}{5}$-sedentary at any vertex. Since both are periodic, Lemma~\ref{tight} implies that both are tightly sedentary. Invoking Theorem~\ref{rook}(2), we get $\min_{t>0}\lvert U_{Y}(t)_{w,w}\rvert=\frac{1}{5}$ is attained at $t_1=\pi$, and so $Y$ is tightly $\frac{1}{5}$-sedentary at any vertex. But since $\nu_2(3)\neq \nu_2(4)$, we cannot say that $X$ is tightly $\frac{1}{6}$-sedentary. Indeed, by computing $U_{K_3}(t)$ and $U_{K_4}(t)$, and using the fact that $\lvert U_{X}(t)_{w,w}\rvert=\lvert U_{K_3}(t)_{u,u}\rvert\cdot \lvert U_{K_4}(t)_{v,v}\rvert$, where $w=(u,v)$, one checks that $\min_{t>0}\lvert U_{X}(t)_{w,w}\rvert \approx 0.2064$ is attained at $t_1\approx 0.9556$. Thus, $X$ is tightly $C$-sedentary at any vertex, where $C\approx 0.7936$. Moreover, since $\min_{t>0}\lvert U_{K_3}(t)_{u,u}\rvert =\frac{1}{3}$ and $\min_{t>0}\lvert U_{K_4}(t)_{v,v}\rvert=\frac{1}{2}$, which are attained at $t_1=\frac{\pi}{3}$ and $t_1=\frac{\pi}{4}$ resp., we conclude that the converse of Theorem~\ref{Cart1}(3) does not hold.
\end{example}

We now prove the main result in this section which could be used to prove that a vertex is sedentary.

\begin{theorem}
\label{eureka}
Let $u$ be a vertex of $X$ with $\sigma_u(M)=\{\lambda_1,\ldots,\lambda_r\}$, where $E_j$ is the orthogonal projection matrix corresponding to $\lambda_j$. If $S$ is a non-empty proper subset of $\sigma_u(M)$, say $S=\{\lambda_1,\ldots,\lambda_s\}$, such that
\begin{equation}
\label{eureka1}
\sum_{j=1}^s(E_j)_{u,u}=a
\end{equation}
for some $\frac{1}{2}\leq a<1$, then
\begin{equation}
\label{eureka2}
\lvert U_M(t)_{u,u}\rvert\geq \left\lvert\sum_{j=1}^s e^{it\lambda_j}(E_j)_{u,u}\right\rvert- (1-a)\quad \text{for all}\ t.
\end{equation}
If there exists a time $t_1>0$ such that
\begin{equation}
\label{eureka3}
\left\lvert\sum_{j=1}^s e^{it_1\lambda_j}(E_j)_{u,u}\right \rvert\geq 1-a,
\end{equation}
and for all $j\in\{1,\ldots,s\}$ and $k\in\{s+1,\ldots,r\}$,
\begin{equation}
\label{eureka333}
e^{it_1(\lambda_1-\lambda_j)}=1\quad  \text{and}\quad e^{it_1(\lambda_{1}-\lambda_{k})}=-1,
\end{equation}
then equality holds in (\ref{eureka2}), in which case $\lvert U_M(t_1)_{u,u} \rvert=2a-1$ and $u$ is periodic at time $2t_1$.
\end{theorem}

\begin{proof}
For brevity, let $\alpha_j=(E_{j})_{u,u}$ for each $j=1,\ldots,r$. We know that $U_M(t)_{u,u}=\sum_{j=1}^r\alpha_je^{it\lambda_j}$. Suppose (\ref{eureka1}) holds, where $1\leq s<r$ and $\frac{1}{2} \leq a<1$. Then $\sum_{j=s+1}^r\alpha_j=1-a$, and because $\alpha_j>0$ for each $j$, we obtain $\left \lvert\sum_{k=s+1}^r \alpha_ke^{it\lambda_k}\right \rvert\leq \sum_{k=s+1}^r \alpha_k=1-a$ by triangle inequality. Hence, for all $t$, we have
\begin{center}
$\lvert U_M(t)_{u,u}\rvert\stackrel{(*)}{\geq} \left\lvert\sum_{j=1}^s \alpha_je^{it\lambda_j}\right \rvert-\left\lvert\sum_{k=s+1}^r \alpha_ke^{it\lambda_k}\right \rvert \stackrel{(**)}{\geq} \left\lvert\sum_{j=1}^s \alpha_je^{it\lambda_j}\right\rvert-(1-a)$.
\end{center}
This proves (\ref{eureka2}). Equality holds in $(**)$ if and only if for some $t_1$ and $\gamma\in\mathbb{C}$, $e^{it_1\lambda_k}=-\gamma$ for each $k\in\{s+1,\ldots,r\}$. This reduces $(*)$ to $\left \lvert\sum_{j=1}^s \alpha_je^{it_1(\lambda_j-\lambda_{s+1}-\pi)}-(1-a)\right\rvert\geq\ \left \lvert\sum_{j=1}^s \alpha_je^{it_1\lambda_j}\right \rvert-(1-a)$, which is an equality if and only if $\left\lvert\sum_{j=1}^s \alpha_je^{it_1\lambda_j}\right \rvert\geq 1-a$ and $e^{it_1(\lambda_j-\lambda_{s+1})}=-1$ for each $j$. The latter yields $e^{it_1\lambda_j}=\gamma$ for $j=1,\ldots,r$. This proves (\ref{eureka3}) and (\ref{eureka333}). If equality holds in $(*)$ and $(**)$, then $\left\lvert\sum_{j=1}^s \alpha_je^{it_1\lambda_j}\right\rvert=a$, and so $\lvert U_M(t_1)_{u,u}\rvert=2a-1$. The statement about periodicity is straightforward.
\end{proof}

The following lemma helps us identify sharply sedentary vertices which are not tightly sedentary.

\begin{lemma}
\label{eurekarem2}
Suppose the premise of Theorem~\ref{eureka} holds. If $\ell_j$ and $m_j$ are integers such that
\begin{equation*}
\label{eur}
\sum_{j=1}^sm_j\lambda_j+\sum_{j=s+1}^r\ell_j\lambda_j=0\quad \text{and}\quad \sum_{j=1}^sm_j+\sum_{j=s+1}^r\ell_j=0
\end{equation*}
implies that $\sum_{j=1}^sm_j$ is even, then there exists a sequence $\{t_k\}$ such that $\lim_{k\rightarrow\infty}\lvert U_M(t_k)_{u,u} \rvert=2a-1$.
\end{lemma}

The proof of Lemma~\ref{eurekarem2} is similar to the proof of a characterization of PGST between two vertices \cite[Lemma 2.2]{Kempton2017a}, except that we replace the sets $\sigma_{uv}^+(M)$ and $\sigma_{uv}^-(M)$ resp. by $\sigma_u(M)\backslash S$ and $S$.

Using Theorem~\ref{eureka} and Lemma~\ref{eurekarem2}, we obtain the following sufficient conditions for sedentariness.

\begin{corollary}
\label{singleton}
Let  $u$ be a vertex of $X$ and suppose $\varnothing\neq S\subseteq \sigma_u(M)$.
\begin{enumerate}
\item Let $S=\{\lambda_1\}$. If $(E_1)_{u,u}=a$, then $\lvert U(t)_{u,u}\rvert\geq 2a-1$ for all $t$. The following also hold.
\begin{enumerate}
\item If $a>\frac{1}{2}$, then $u$ is $(2a-1)$-sedentary. This is tight (resp., sharp) whenever (\ref{eureka333}) (resp., Lemma~\ref{eurekarem2}) holds. Moreover, if $u$ is periodic, then $u$ is tightly $C$-sedentary for some $C\geq 2a-1$.
\item Suppose (\ref{eureka3}) and (\ref{eureka333}) hold, or Lemma~\ref{eurekarem2} holds. If $a=\frac{1}{2}$, then $u$ is not sedentary.
\end{enumerate}
\item Let $\lvert S\rvert\geq 2$, $b>0$ and $F(t)=\left\lvert\sum_{j=1}^s e^{it\lambda_j}(E_{j})_{u,u}\right\rvert$. If $a>\frac{1}{2}$, then $u$ is sedentary whenever (i) $F(t)-(1-a)>b$ for all $t$ or (ii) $F(t)\geq 1-a$ for all $t$, $u$ is periodic and $U(t_1)_{u,u}\neq 0$ for all $t_1$ with $F(t_1)=1-a$.
\end{enumerate}
\end{corollary}

\begin{proof}
The statement in (1) follows from Theorem~\ref{eureka}. To prove (1a), let $a<\frac{1}{2}$. Then $u$ is clearly $2a$-sedentary. Since (\ref{eureka3}) holds by default, the sedentariness is tight by Theorem~\ref{eureka} whenever (\ref{eureka333}) holds. If the premise of Lemma~\ref{eurekarem2} holds, then $\inf_{t>0}\lvert U_M(t)_{u,u}\rvert=2a-1$, and so $u$ is tightly sedentary. If we add that $u$ is periodic, then $\min_{t>0}\lvert U_M(t)_{u,u}\rvert=C\geq 2a-1$, where $\min_{t>0}\lvert U_M(t)_{u,u}\rvert$ is attained at some $t_1\in(o,\rho)$. Thus, $u$ is $C$-sedentary. For (1b), if $a=\frac{1}{2}$, then $\lvert U_M(t)_{u,u}\rvert\geq 0$. If (\ref{eureka3}) and (\ref{eureka333}) hold, then $\lvert U_M(t_1)_{u,u}\rvert=0$ at some $t_1\in(o,\rho)$, while if Lemma~\ref{eurekarem2} holds, then $\inf_{t>0}\lvert U_M(t)_{u,u}\rvert=0$. This proves (1b). Finally, let $\lvert S\rvert\geq 2$ and $a> \frac{1}{2}$. If (2i) holds, then $u$ is sedentary by (\ref{eureka2}). If (2ii) holds, then $\lvert U_M(t)_{u,u}\rvert>0$ for all $t$, and so Lemma~\ref{tight} implies that $u$ is tightly sedentary. This proves (2).
\end{proof}

As we will see, Corollary~\ref{singleton}(1) will be useful in the later sections. We end this section with an example that illustrates Corollary~\ref{singleton}.

\begin{example}
Consider the path $P_3$ with end vertex $u$. Then $\sigma_u(A)=\{\pm\sqrt{2},0\}$ with associated eigenvectors $(1,\pm\sqrt{2},1)$ and $(1,0,-1)$, while $\sigma_u(L)=\{3,1,0\}$ with associated eigenvectors $(1,-2,1)$, $(1,0,-1)$ and $\textbf{1}$. Note that $u$ is periodic in both cases. Moreover,
\begin{center}
$U_A(t)_{u,u}=\frac{1}{4}e^{i\sqrt{2}t}+\frac{1}{4}e^{-i\sqrt{2}t}+\frac{1}{2}\quad \text{and}\quad U_L(t)_{u,u}=\frac{1}{6}e^{i3t}+\frac{1}{2}e^{it}+\frac{1}{3}$.
\end{center}
For $A$, let $S=\{0\}$. Then $(E_{0})_{u,u}=1/2$ and one checks that (\ref{eureka3}) and (\ref{eureka333}) hold at $t_1=\pi/\sqrt{2}$. By Corollary~\ref{singleton}(1b), $u$ is not sedentary, which is consistent with the fact that adjacency PST occurs between end vertices of $P_3$ at $t_1$. For $L$, take $S=\{3,1\}$ so that $(E_{3})_{u,u}+(E_{1})_{u,u}=2/3$. Applying Theorem~\ref{eureka} with $a=2/3$, we get that $F(t)=\left\lvert\frac{1}{6}e^{i3t}+\frac{1}{2}e^{it}\right\rvert-\frac{1}{3}$. Now, $F(t_1)=0$ if and only if $t_1=j\pi/2$ for any odd $j$. Since $ U_L(t_1)_{ u,u}\neq 0$ and $u$ is periodic, Corollary~\ref{singleton}(2ii) implies that $u$ is tightly sedentary. 
\end{example}

\section{Twin vertices}\label{secTwins}

In this section, we show that a vertex with at least two twins is sedentary. Unless otherwise stated, all results in this section apply to both $\textbf{A}$ and $\mathcal{A}$.

Two vertices $u$ and $v$ of $X$ are \textit{twins} if (i) $N_X(u)\backslash \{u,v\}=N_X(v)\backslash \{u,v\}$, (ii) the edges $(u,w)$ and $(v,w)$ have the same weight for each $w\in N_X(u)\backslash \{u,v\}$, and (iii) the loops on $u$ and $v$ have the same weight if they exist.  We say that a subset $T=T(\omega,\eta)$ of $V(X)$ with at least two vertices is a \textit{set of twins} in $X$ if each pair of vertices in $T$ are twins, where each vertex in $T$ has a loop of weight $\omega$ whenever $\omega\neq 0$ and every pair of vertices in $T$ are connected by an edge with weight $\eta$ whenever $\eta\neq 0$. Since there exists an automorphism that switches any pair of twins, it follows that all vertices in $T$ are pairwise cospectral. For a more extensive treatment of the role of twin vertices in quantum state transfer, see \cite{Monterde}.

We now restate a spectral characterization of twin vertices \cite[Lemma 2.9]{MonterdeELA}.

\begin{lemma}
\label{alphabeta}
Let $T=T(\omega,\eta)$ be a set of twins in $X$. Then $u,v\in T$ if and only if
$\textbf{e}_u-\textbf{e}_v$ is an eigenvector of $M$ corresponding to the eigenvalues $\theta$ given by
\begin{equation}
\label{adjalpha}
\theta=
\begin{cases}
 \alpha \operatorname{deg}(u)+\omega-\eta, &\text{if $M=\textbf{A}$}\\
 \frac{\omega-\eta}{\operatorname{deg}(u)}, &\text{if $M=\mathcal{A}$}.
\end{cases}
\end{equation}
\end{lemma}

If $u\in T$, then $\theta\in\sigma_u(M)$ by Lemma~\ref{alphabeta}. We use this to prove our main result.

\begin{theorem}
\label{sed}
Let $T$ be a set of twins in $X$. If $u\in T$ with $\sigma_u(M)=\{\theta,\lambda_2,\ldots,\lambda_r\}$, then
\begin{equation}
\label{UHH}
\lvert U_M(t)_{u,u}\rvert \geq1- \frac{2}{\lvert T\rvert}\quad \text{for all $t$},
\end{equation}
with equality whenever (\ref{eureka333}) holds with $S\hspace{-0.01in} =\{\theta\}$. Further, if $\lvert T\rvert\geq 3$, then $u$ is $(1-\frac{2}{\lvert T\rvert})$-sedentary.
\end{theorem}
\begin{proof}
Let $T$ be a set of twins in $X$. If we index the first $\lvert T\rvert$ rows of $M$ by the elements of $T$, then for a fixed $u\in T$, Lemma~\ref{alphabeta} implies that $\textbf{e}_u-\textbf{e}_v$ is an eigenvector for $M$ for all $v\in T\backslash\{u\}$ corresponding to the eigenvalue $\theta$ in (\ref{adjalpha}). Assuming $u$ is the first row of $M$, we get $E_{\theta}=\left(I_{\lvert T\rvert}-\frac{1}{\lvert T\rvert}\textbf{J}_{\lvert T\rvert}\right)+F$ for some matrix $F$. Taking $S=\{\theta\}$, we get $1-a=1-\frac{1}{\lvert T\rvert}\geq \frac{1}{2}$. Applying Corollary~\ref{singleton}(1) yields the desired result.
\end{proof}

\begin{remark}
\label{complement}
If $X$ is simple and unweighted, and $T$ is a set of twins in $X$, then $T$ is also a set of twins in the complement $X^c$ of $X$. Thus, if $X^c$ is connected, then Theorem~\ref{sed} also holds for $X^c$.
\end{remark}

Theorem~\ref{sed} reveals that twin vertices in quantum walks behave like vertices in a complete graph, which is an interesting observation because the underlying graph induced by a set of twins is either complete or empty. But unlike complete graphs, equality in (\ref{UHH}) may not be attained for other graphs.

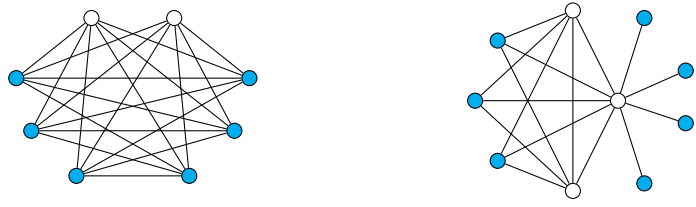
\begin{figure}[h!]
	\begin{center}
		\begin{tikzpicture}
		\tikzset{enclosed/.style={draw, circle, inner sep=0pt, minimum size=.2cm}}
	 
		\node[enclosed, fill=cyan] (v_1) at (-0.2,1.3) {};
		\node[enclosed, fill=cyan] (v_2) at (0,0.6) {};
		\node[enclosed, fill=cyan] (v_4) at (0.6,0) {};
		
		\node[enclosed] (v_3) at (0.8,2.1) {};
		\node[enclosed] (v_5) at (1.9,2.1) {};
		
		\node[enclosed, fill=cyan] (v_6) at (2.1,0) {};
		\node[enclosed, fill=cyan] (v_7) at (2.7,0.6) {};
		\node[enclosed, fill=cyan] (v_8) at (2.9,1.3) {};
		
		\draw (v_1) -- (v_3);
		\draw (v_1) -- (v_5);
		\draw (v_2) -- (v_3);
		\draw (v_2) -- (v_5);
		\draw (v_4) -- (v_3);
		\draw (v_4) -- (v_5);
		\draw (v_6) -- (v_3);
		\draw (v_6) -- (v_5);
		\draw (v_7) -- (v_3);
		\draw (v_7) -- (v_5);
		\draw (v_8) -- (v_3);
		\draw (v_8) -- (v_5);
		\draw (v_1) -- (v_6);
		\draw (v_1) -- (v_7);
		\draw (v_1) -- (v_8);
		\draw (v_2) -- (v_6);
		\draw (v_2) -- (v_7);
		\draw (v_2) -- (v_8);
		\draw (v_4) -- (v_6);
		\draw (v_4) -- (v_7);
		\draw (v_4) -- (v_8);
 		
 		\node[enclosed,fill=cyan] (w_1) at (6.2,0.2) {};
 		\node[enclosed,fill=cyan] (w_2) at (5.9,1) {};
 		\node[enclosed,fill=cyan] (w_3) at (6.2,1.8) {};
 		\node[enclosed] (w_4) at (7.2,-0.2) {};
 		\node[enclosed] (w_5) at (7.2,2.2) {};
 		\node[enclosed] (w_6) at (7.8,1) {};
 		\node[enclosed,fill=cyan] (w_7) at (8.15,-0.1) {};
 		\node[enclosed,fill=cyan] (w_8) at (8.7,0.7) {};
 		\node[enclosed,fill=cyan] (w_9) at (8.7,1.4) {};
 		\node[enclosed,fill=cyan] (w_10) at (8.15,2.1) {};
 		\draw (w_4) -- (w_5);
 		\draw (w_1) -- (w_4);
 		\draw (w_2) -- (w_4);
 		\draw (w_3) -- (w_4);
		\draw (w_1) -- (w_5);
 		\draw (w_2) -- (w_5);
 		\draw (w_3) -- (w_5);
 		\draw (w_1) -- (w_6);
 		\draw (w_2) -- (w_6);
 		\draw (w_3) -- (w_6);
 		\draw (w_4) -- (w_6);
 		\draw (w_5) -- (w_6);
 		\draw (w_7) -- (w_6);
 		\draw (w_8) -- (w_6);
 		\draw (w_9) -- (w_6);
 		\draw (w_10) -- (w_6);
		\end{tikzpicture}
	\end{center}
	\caption{The complete multipartite graph $K_{2,3,3}$ (left) and the threshold graph $((O_3\vee K_2)\cup O_4)\vee K_1$ (right) with sedentary vertices marked blue}\label{yay1}
\end{figure}

\subsection*{Joins}

Since the property of being twins is preserved under joins, Theorem~\ref{sed} yields the following results.

\begin{corollary}
\label{corjoin}
Let $T$ be a set of twins in $Y$. If $\lvert T\rvert\geq 3$, then the vertices in $T$ are $(1-\frac{2}{\lvert T\rvert})$-sedentary in $Y\vee X$ for any weighted graph $X$ with possible loops. 
\end{corollary}

\begin{corollary}
\label{corjoin1}
Let $X$ be a weighted graph with possible loops. For each $m\geq 3$, the vertices of $K_m$ and $O_m$ resp. are $(1-\frac{2}{m})$-sedentary in $K_m\vee X$ and $O_m\vee X$.
\end{corollary}

By Corollary~\ref{corjoin1}, a degree $m-1$ vertex of $K_m\backslash e=K_{m-2}\vee O_2$ is $(1-\frac{2}{m-1})$-sedentary for all $m\geq 5$.

We now examine sedentariness in two well known classes of graphs obtained using the join operation.

\begin{corollary}
\label{cmgtg}
Let $n_1,\ldots, n_{2k}$ be integers such that $n_j\geq 3$ for some $j\in\{1,\ldots,k\}$.
\begin{enumerate}
\item Each vertex of $K_{n_1,n_2,\ldots,n_k}$ in the partite set of size $n_j$ is $(1-\frac{2}{n_j})$-sedentary. Moreover, if $\lvert\{\ell:n_{\ell}=1\}\rvert=p\geq 3$, then each vertex in a singleton partite set of $K_{n_1,n_2,\ldots,n_k}$ is $(1-\frac{2}{p})$-sedentary.
\item Each vertex of $Z\in\{K_{n_j},O_{n_j}\}$ is $(1-\frac{2}{n_j})$-sedentary in the threshold graph
\begin{equation}
\label{thgr}
((((O_{n_1}\vee K_{n_2})\cup O_{n_3})\vee K_{n_4})\cdots)\vee K_{n_{2k}}\quad\text{or} \quad ((((K_{n_1}\cup O_{n_2})\vee K_{n_3})\cup O_{n_4})\cdots)\vee K_{n_{2k+1}}.
\end{equation}
\end{enumerate}
\end{corollary}

\begin{proof}
If $n_j\geq 3$, then each partite set of $K_{n_1,n_2,\ldots,n_k}$ and those vertices in each $K_{n_j}$ and $O_{n_j}$ in (\ref{thgr}) form a set of twins. If $\lvert \{\ell:n_{\ell}=1\}\rvert=p\geq 3$, then the singleton partite sets also form a set of twins size $p$. Applying Theorem~\ref{sed} yields the desired result.
\end{proof}

Threshold graphs with form given in (\ref{thgr}), where $n_1\geq 2$ and $n_j\geq 1$ for $j\geq 2$, are precisely all the connected threshold graphs as characterized by Kirkland and Severini (see \cite[Lemma 1]{Severini}).

Next, we have the following immediate consequence of Theorem~\ref{sed}.

\begin{corollary}
\label{sedF}
Let $\mathscr{F}$ be a family of graphs with a set of twins $T$ with $\lvert T\rvert\geq 3$.
\begin{enumerate}
\item If $\lvert T\rvert$ is fixed for all $X\in\mathscr{F}$, then $\mathscr{F}$ is $(1-\frac{2}{\lvert T\rvert})$-sedentary at every vertex in $T$.
\item If $\lvert V(X)\backslash T\rvert$ is fixed for all $X\in\mathscr{F}$, then $\mathscr{F}$ is sedentary at every vertex of $T$.
\end{enumerate}
\end{corollary}

For the family $\mathscr{K}_1$ of complete graphs on $m\geq 5$ vertices minus an edge, all vertices in $K_{m}\backslash e$, except for the non-adjacent pair, form a set of twins $T$ with $\lvert T\rvert=m-2$. Thus, $\lvert V(X)\backslash T\rvert=2$ is fixed, and so by Corollary~\ref{sedF}(2), $\mathscr{K}_1$ is a family that is sedentary at all vertices except for the non-adjacent pair.

The next result follows immediately from Corollaries \ref{corjoin1}, \ref{cmgtg} and \ref{sedF}.

\begin{corollary}
\label{sedg2}
The following hold.
\begin{enumerate}
\item Let $\mathscr{F}_1$ and $\mathscr{F}_2$ be families of graphs resp.\ of the form $O_m\vee X$ and $K_m\vee X$. Let $Z\in\{O_m,K_m\}$. If $X$ has fixed number of vertices, then each $\mathscr{F}_i$ is sedentary at every vertex of $Z$. If $m\geq 3$ is fixed, then each $\mathscr{F}_i$ is $(1-\frac{2}{m})$-sedentary at every vertex of $Z$.
\item Let $\mathscr{F}_1$ and $\mathscr{F}_2$ resp. be families of complete multipartite graphs $K_{n_1,\ldots,n_k}=\bigvee_{\ell=1}^kO_{n_{\ell}}$ and threshold graphs in (\ref{thgr}). Let $Z\in \{K_{n_j},O_{n_j}\}$. If $n_j\geq 3$ is fixed, then each $\mathscr{F}_i$ is $(1-\frac{2}{n_j})$-sedentary at every vertex in $Z$. If $k\geq 1$ is fixed and $n_j\rightarrow\infty$, then each $\mathscr{F}_i$ is sedentary at every vertex of $Z$.
\end{enumerate}
\end{corollary}

\begin{figure}[h!]
	\begin{center}
		\begin{tikzpicture}
		\tikzset{enclosed/.style={draw, circle, inner sep=0pt, minimum size=.2cm}}
	 
		\node[enclosed] (v_7) at (0,1) {};
		\node[enclosed] (v_8) at (0.7,1) {};
		\node[enclosed] (v_1) at (1.4,1) {};
		\node[enclosed, fill=cyan] (v_2) at (2.1,0.1) {};
		\node[enclosed, fill=cyan] (v_3) at (2.1,1.9) {};
		\node[enclosed, fill=cyan] (v_4) at (3,0.1) {};
		\node[enclosed, fill=cyan] (v_5) at (3,1.9) {};
		\node[enclosed, fill=cyan] (v_6) at (3.7,1) {};
		
		\draw (v_1) -- (v_2);
		\draw (v_1) -- (v_3);
		\draw (v_1) -- (v_4);
		\draw (v_1) -- (v_5);
		\draw (v_1) -- (v_6);
		\draw (v_2) -- (v_3);
		\draw (v_2) -- (v_4);
		\draw (v_2) -- (v_5);
		\draw (v_2) -- (v_6);
		\draw (v_3) -- (v_4);
		\draw (v_3) -- (v_5);
		\draw (v_3) -- (v_6);
		\draw (v_4) -- (v_5);
		\draw (v_4) -- (v_6);
		\draw (v_5) -- (v_6);
		\draw (v_8) -- (v_1);
	    \draw (v_7) -- (v_8);
 		
 		\node[enclosed] (w_9) at (5,1) {};
 		\node[enclosed] (w_8) at (5.7,1) {};
 		\node[enclosed, fill=magenta] (w_7) at (6.4,1) {};
 		\node[enclosed, fill=cyan] (w_1) at (7.2,1) {};
		\node[enclosed, fill=cyan] (w_2) at (7.9,0.1) {};
		\node[enclosed, fill=cyan] (w_4) at (7.9,1.9) {};
		\node[enclosed, fill=magenta] (w_3) at (8.6,1) {};
		\node[enclosed, fill=magenta] (w_11) at (8.6,1.9) {};
		\node[enclosed, fill=magenta] (w_12) at (8.6,0.1) {};
		\node[enclosed] (w_5) at (9.3,1) {};
		\node[enclosed] (w_13) at (9.3,0.1) {};
		\node[enclosed] (w_15) at (9.3,1.9) {};
		\node[enclosed] (w_6) at (10,1) {};
		\node[enclosed] (w_14) at (10,0.1) {};
		\node[enclosed] (w_16) at (10,1.9) {};
		\draw (w_1) -- (w_2);
		\draw (w_1) -- (w_3);
		\draw (w_4) -- (w_2);
 		\draw (w_2) -- (w_3);
 		\draw (w_4) -- (w_3);
 		\draw (w_4) -- (w_1);
 		\draw (w_5) -- (w_3);
 		\draw (w_5) -- (w_6);
 		\draw (w_1) -- (w_7);
 		\draw (w_8) -- (w_7);
 		\draw (w_2) -- (w_7);
 		\draw (w_4) -- (w_7);
 		\draw (w_8) -- (w_9);
 		\draw (w_11) -- (w_1);
 		\draw (w_11) -- (w_2);
 		\draw (w_11) -- (w_4);
 		\draw (w_12) -- (w_1);
 		\draw (w_12) -- (w_2);
 		\draw (w_12) -- (w_4);
 		\draw (w_12) -- (w_13);
 		\draw (w_13) -- (w_14);
 		\draw (w_11) -- (w_15);
 		\draw (w_15) -- (w_16);
 		
 		\node[enclosed] (u_1) at (11.3,0.4) {};
 		\node[enclosed] (u_2) at (11.3,1.6) {};
 		\node[enclosed] (u_3) at (12,0.4) {};
 		\node[enclosed] (u_4) at (12,1.6) {};
 		\node[enclosed, fill=magenta] (u_5) at (12.7,0.4) {};
 		\node[enclosed, fill=magenta] (u_6) at (12.7,1.6) {};
 		\node[enclosed,fill=cyan] (u_7) at (13.6,-0.1) {};
 		\node[enclosed, fill=magenta] (u_8) at (13.3,1) {};
 		\node[enclosed,fill=cyan] (u_9) at (13.6,2.1) {};
 		\node[enclosed,fill=cyan] (u_10) at (14.6,0.6) {};
 		\node[enclosed] (u_11) at (14.4,1) {};
 		\node[enclosed,fill=cyan] (u_12) at (14.6,1.4) {};
 		\node[enclosed] (u_13) at (15.4,1) {};
		\draw (u_1) -- (u_3);
		\draw (u_5) -- (u_3);
		\draw (u_2) -- (u_4);
		\draw (u_4) -- (u_6);
		\draw (u_5) -- (u_6);
		\draw (u_8) -- (u_6);
		\draw (u_5) -- (u_8);
		\draw (u_5) -- (u_7);
		\draw (u_5) -- (u_9);
		\draw (u_5) -- (u_10);
		\draw (u_5) -- (u_12);
		\draw (u_6) -- (u_7);
		\draw (u_6) -- (u_9);
		\draw (u_6) -- (u_10);
		\draw (u_6) -- (u_12);
		\draw (u_8) -- (u_7);
		\draw (u_8) -- (u_9);
		\draw (u_8) -- (u_10);
		\draw (u_8) -- (u_12);
		\draw (u_8) -- (u_11);
		\draw (u_13) -- (u_11);
		
		\end{tikzpicture}
	\end{center}
	\caption{The lollipop graph $L_{4,2}$ (left), the graph $X_{3,4,2}$ (center) and the graph $Y_{3,4,2}$ (right) with sedentary vertices marked blue}\label{yay}
\end{figure}
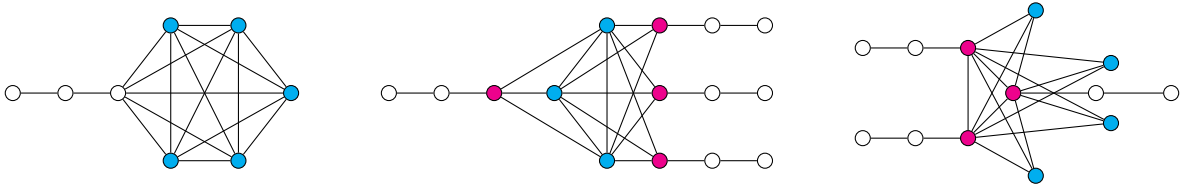

\subsection*{Graphs with tails}

For $n\geq 4$ and $k\geq 1$, let $L_{n,k}$ be a lollipop graph, which is a graph obtained after attaching a path $P_k$ to a vertex $u$ of $K_n$. The vertices $v\neq u$ of $K_n$ in $L_{n,k}$ form a set of twins of size $n-1\geq 3$, and so each of them is $(1-\frac{2}{n-1})$-sedentary by Theorem~\ref{sed}. More generally, if $n\geq k+3$, then attaching $k$ paths (possibly with different lengths) to $k$ vertices of $K_n$ leaves the remaining $n-k$ vertices of $K_n$ $(1-\frac{2}{n-k})$-sedentary in the resulting graph. The same holds in the complement of $L_{n,k}$.

For $n,m\geq 3$ and $k\geq 0$, let $X_{n,m,k}$ and $Y_{n,m,k}$ be graphs obtained from $K_n\vee O_m$ by attaching copies of $P_k$ resp. to the vertices of $O_m$ and $K_n$. The vertices of $K_n$ form a set of twins in $X_{n,m,k}$ of size $n$, while those of $O_m$ form a set of twins in $Y_{n,m,k}$ of size $m$. Thus, the vertices of $K_n$ and $O_m$ resp. are $(1-\frac{2}{n})$- and $(1-\frac{2}{m})$-sedentary in $X_{n,m,k}$ and $Y_{n,m,k}$. This remains true even if we vary the lengths of paths attached to $O_m$ and $K_n$. This also holds in the complements of $X_{n,m,k}$ and $Y_{n,m,k}$.

The above considerations combined with Corollary~\ref{sedF} yield the following result.

\begin{corollary}
\label{sedg1}
Let $\mathscr{F}$ be a family of simple unweighted lollipop graphs $L_{n,k}$, $\mathscr{F}_1$ be a family of graphs $X_{n,m,k}$, and $\mathscr{F}_2$ be a family of graphs $Y_{n,m,k}$.
\begin{enumerate}
\item Let $n\geq 4$ and $k\geq 1$. If $k$ is fixed, then $\mathscr{F}$ is sedentary at every vertex of $v\neq u$ of $K_n$. If $n$ is fixed, then $\mathscr{F}$ is $(1-\frac{2}{n-1})$-sedentary at every vertex $v\neq u$ of $K_n$
\item Let $n,m\geq 3$ and $k\geq 0$. If $m$ and $k$ are fixed, then $\mathscr{F}_1$ (resp., $\mathscr{F}_2$) is sedentary at every vertex of $K_n$ (resp., $O_n$). If $n$ is fixed, then $\mathscr{F}_1$ is $(1-\frac{2}{n})$-sedentary at every vertex of $K_n$, while if $m$ is fixed, then $\mathscr{F}_2$ is $(1-\frac{2}{m})$-sedentary at every vertex of $O_m$.
\end{enumerate}
\end{corollary}

For $n\geq 4$, let $L_n$ be an infinite lollipop, which is the resulting graph after attaching an infinite path to a vertex $u$ of a complete graph $K_n$. In \cite[Proposition 3]{Tamon}, Bernard et al.\ showed that the family of infinite lollipops is sedentary at each vertex $v\neq u$ of $K_n$. This complements Corollary~\ref{sedg1}(1) which states that the family of lollipop graphs $L_{n,k}$ with $k$ fixed is sedentary at every vertex $v\neq u$ of $K_n$. They also showed that attaching infinite paths to the vertices of $O_m$ in $K_n\vee O_m$ yields a family that is sedentary at every vertex of $K_n$ \cite[Theorem 4]{Tamon}, which again, complements our result in Corollary~\ref{sedg1}(2a), which states that the family of graphs $X_{n,m,k}$ is sedentary at every vertex of $K_n$ whenever $m$ and $k$ are fixed.

Similar to lollipop graphs, one may construct barbell-type graphs with sedentary vertices. Barbell-type graphs are obtained by joining corresponding vertices of two copies of complete graphs with a path. For instance, if $m,n\geq 4$ and $k\geq 1$ then the barbell-type graph $L_{n,k,m}$ formed by joining vertices $u$ of $K_n$ and $w$ of $K_m$ by a path $P_{k}$ is $(1-\frac{2}{n-1})$- and $(1-\frac{2}{m-1})$-sedentary resp.\ at any vertex $v\neq u$ of $K_n$ and $v\neq w$ of $K_m$. One can then derive results about sedentary families of barbell-type graphs similar to Corollary~\ref{sedg1}.

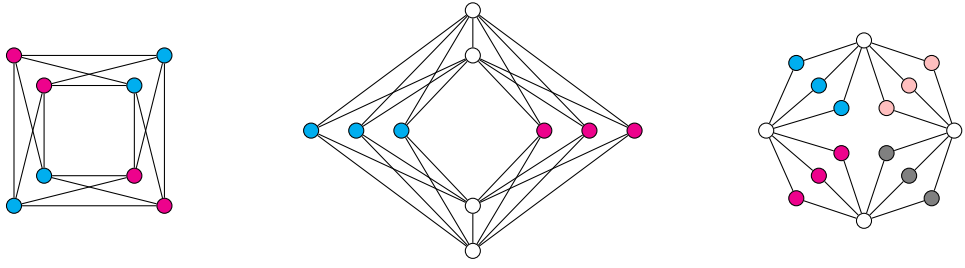
\begin{figure}[h!]
	\begin{center}
		\begin{tikzpicture}
		\tikzset{enclosed/.style={draw, circle, inner sep=0pt, minimum size=.2cm}}
	 
		\node[enclosed,fill=cyan] (v_1) at (0,0) {};
		\node[enclosed,fill=magenta] (v_2) at (0,2) {};
		\node[enclosed,fill=cyan] (v_3) at (0.4,0.4) {};
		\node[enclosed,fill=magenta] (v_4) at (0.4,1.6) {};
		\node[enclosed,fill=magenta] (v_5) at (1.6,0.4) {};
		\node[enclosed,fill=cyan] (v_6) at (1.6,1.6) {};
		\node[enclosed,fill=magenta] (v_7) at (2,0) {};
		\node[enclosed,fill=cyan] (v_8) at (2,2) {};
		
		\draw (v_1) -- (v_4);
		\draw (v_1) -- (v_2);		
		\draw (v_3) -- (v_2);
		\draw (v_4) -- (v_3);
		\draw (v_2) -- (v_6);
		\draw (v_2) -- (v_8);
		\draw (v_4) -- (v_6);
		\draw (v_4) -- (v_8);
		\draw (v_1) -- (v_5);
		\draw (v_1) -- (v_7);
 		\draw (v_3) -- (v_5);
		\draw (v_3) -- (v_7);
		\draw (v_6) -- (v_5);
		\draw (v_6) -- (v_7);
		\draw (v_8) -- (v_5);
		\draw (v_8) -- (v_7);
		
		\node[enclosed,fill=cyan] (w_1) at (3.95,1) {};
		\node[enclosed,fill=cyan] (w_2) at (4.55,1) {};
		\node[enclosed,fill=cyan] (w_3) at (5.15,1) {};
		\node[enclosed] (w_4) at (6.1,2) {};
		\node[enclosed] (w_5) at (6.1,2.6) {};
		\node[enclosed] (w_6) at (6.1,0) {};
		\node[enclosed] (w_7) at (6.1,-0.6) {};
		\node[enclosed,fill=magenta] (w_8) at (7.05,1) {};
		\node[enclosed,fill=magenta] (w_9) at (7.65,1) {};
		\node[enclosed,fill=magenta] (w_10) at (8.25,1) {};
		
		\draw (w_1) -- (w_4);
		\draw (w_1) -- (w_5);
		\draw (w_1) -- (w_6);
		\draw (w_1) -- (w_7);
		\draw (w_2) -- (w_4);
		\draw (w_2) -- (w_5);
		\draw (w_2) -- (w_6);
		\draw (w_2) -- (w_7);
		\draw (w_3) -- (w_4);
		\draw (w_3) -- (w_5);
		\draw (w_3) -- (w_6);
		\draw (w_3) -- (w_7);
		\draw (w_8) -- (w_4);
		\draw (w_8) -- (w_5);
		\draw (w_8) -- (w_6);
		\draw (w_8) -- (w_7);
		\draw (w_9) -- (w_4);
		\draw (w_9) -- (w_5);
		\draw (w_9) -- (w_6);
		\draw (w_9) -- (w_7);
		\draw (w_10) -- (w_4);
		\draw (w_10) -- (w_5);
		\draw (w_10) -- (w_6);
		\draw (w_10) -- (w_7);
		\draw (w_5) -- (w_4);
		\draw (w_6) -- (w_7);
		
		\node[enclosed] (u_1) at (10,1) {};
		\node[enclosed,fill=cyan] (u_2) at (10.4,1.9) {};
		\node[enclosed,fill=magenta] (u_3) at (10.4,0.1) {};
		\node[enclosed,fill=cyan] (u_4) at (10.7,1.6) {};
		\node[enclosed,fill=magenta] (u_5) at (10.7,0.4) {};
		\node[enclosed,fill=cyan] (u_6) at (11,1.3) {};
		\node[enclosed,fill=magenta] (u_7) at (11,0.7) {};
		\node[enclosed] (u_8) at (11.3,2.2) {};
		\node[enclosed] (u_9) at (11.3,-0.2) {};
		\node[enclosed,fill=gray] (u_10) at (11.6,0.7) {};
		\node[enclosed,fill=pink] (u_11) at (11.6,1.3) {};
		\node[enclosed,fill=gray] (u_12) at (11.9,0.4) {};
		\node[enclosed,fill=pink] (u_13) at (11.9,1.6) {};
		\node[enclosed,fill=gray] (u_14) at (12.2,0.1) {};
		\node[enclosed,fill=pink] (u_15) at (12.2,1.9) {};
		\node[enclosed] (u_16) at (12.5,1) {};
		
		\draw (u_1) -- (u_2);
		\draw (u_1) -- (u_3);
		\draw (u_1) -- (u_4);
		\draw (u_1) -- (u_5);
		\draw (u_1) -- (u_6);
		\draw (u_1) -- (u_7);
		\draw (u_8) -- (u_2);
		\draw (u_8) -- (u_4);
		\draw (u_8) -- (u_6);
		\draw (u_9) -- (u_3);
		\draw (u_9) -- (u_5);
		\draw (u_9) -- (u_7);
		\draw (u_16) -- (u_11);
		\draw (u_16) -- (u_13);
		\draw (u_16) -- (u_15);
		\draw (u_16) -- (u_10);
		\draw (u_16) -- (u_12);
		\draw (u_16) -- (u_14);
		\draw (u_9) -- (u_10);
		\draw (u_9) -- (u_12);
		\draw (u_9) -- (u_14);
		\draw (u_8) -- (u_11);
		\draw (u_8) -- (u_13);
		\draw (u_8) -- (u_15);
		
		\end{tikzpicture}
	\end{center}
	\caption{Blow-ups of $C_4$: $C_4^2(V)$ (left), $C_4(2,3,2,3)(V)$ (center), and $C_4^3(E)$ (right) with sets of twins filled with the same color, all members of which are sedentary}\label{yay2}
\end{figure}

\subsection*{Blow-ups}

Let $X$ be a weighted graph with possible loops with vertices $v_1,\ldots,v_n$ and edges with distinct endpoints (i.e., non-loops) $e_1,\ldots,e_m$. Let $(k_1,\ldots,k_n)$ and $(k_1,\ldots,k_m)$ be $n$- and $m$-tuples of positive integers.

A $(k_1,\ldots,k_n)$-\textit{vertex blow-up} of $X$, denoted $X(k_1,\ldots,k_n)(V)$, is the graph obtained by replacing every vertex $v_j$ of $X$ by the graph $X_j\in\{O_{k_j},K_{k_j}\}$ such that a vertex in $X_{j}$ is adjacent to a vertex in $X_{\ell}$ in the resulting graph if and only if $v_j$ and $v_{\ell}$ are adjacent in $X$, and the weight of each edge between $X_j$ and $X_{\ell}$ is the same as the weight of the edge $[v_j,v_{\ell}]$ in $X$. If $k_1=\cdots=k_m=k$, then we call the resulting graph a \textit{$k$-vertex blow-up} of $X$, denoted $X^{k}(V)$. Vertex blow-ups in the literature typically mean replacing each vertex by an empty graph, but in our definition, we have the freedom to choose between an empty or a complete graph. For example, $K_{m,n}$ and $K_{m+n}$ are $(m,n)$-vertex blow-ups of $K_2$, where each vertex of $K_2$ was replaced by an empty graph for the former, and by a complete graph for the latter.

A \textit{$(k_1,\ldots,k_m)$-edge blow-up} of $X$, denoted $X(k_1,\ldots,k_n)(E)$, is a graph obtained by replacing every edge $e_j=[u_j,v_j]$ of $X$ by $X_j\in\{O_{k_j},K_{k_j}\}$ and adding edges $[u_j,w]$ and $[v_j,w]$ for all vertices $w$ of $X_j$, each with weight equal to that of $e_j$. If $k_1=\cdots=k_m=k$, then we call the resulting graph a \textit{$k$-edge blow-up} of $X$, denoted $X^{k}(E)$. A 1-edge blow-up of $X$ is obtained by subdividing every edge of $X$.

\begin{theorem}
\label{bup}
Let $X$ be a weighted graph with possible loops with vertices $v_1,\ldots,v_n$ and edges $e_1,\ldots,e_m$ with distinct endpoints. Let $(k_1,\ldots,k_n)$ and $(k_1,\ldots,k_m)$ be $n$ and $m$-tuples of positive integers.
\begin{enumerate}
\item If $k_j\geq 3$ for some $j$, then the vertices of $X_j\in\{O_{k_j},K_{k_j}\}$ added in place of $v_j$ (resp., $e_j$) are $(1-\frac{2}{k_j})$-sedentary in $X(k_1,\ldots,k_n)(V)$ (resp., $X(k_1,\ldots,k_m)(E)$). If $k\geq 3$, then each vertex in $X^{k}(V)$ is $(1-\frac{2}{k})$-sedentary, while each vertex in $\bigcup_{j=1}^m X_j$ is $(1-\frac{2}{k})$-sedentary in $X^{k}(E)$.
\item Let $T$ be a set of twins in $X$ with $k_j\geq 2$ for some $v_j\in T$. Suppose $W_1=\bigcup_{v_j\in T,X_j=K_{K_j}}X_j$ and $W_2=\bigcup_{v_j\in T,X_j=O_{k_j}} X_j$. If the vertices in $T$ are pairwise adjacent, then $W_1$ is a set of twins in $X(k_1,\ldots,k_n)(V)$. Otherwise, $W_2$ is a set of twins in $X(k_1,\ldots,k_n)(V)$. If $\lvert W_i\rvert\geq 3$ for some $i\in\{1,2\}$, then each vertex in $W_i$ is $(1-\frac{2}{\lvert W_i\rvert})$-sedentary.
\end{enumerate}
\end{theorem}

\begin{proof}
Since the vertices of $X_j$ form a set of twins of size $k_j\geq 3$, (1) follows directly from Theorem~\ref{sed}. Now, let $T$ be a set of twins in $X$ such that $k_j\geq 2$ for some $v_j\in T$. Note that the vertices in $T$ are either all pairwise adjacent, or all pairwise non-adjacent. Suppose the former holds. If $v_1,v_2\in T$ are distinct, and we replace $v_1$ by $O_{k_1}$ with $k_1\geq 2$ and $v_2$ by $Z_2\in\{O_{k_2},K_{k_2}\}$, then a vertex $u_1$ in $O_{k_1}$ and a vertex $u_2$ in $Z$ are not twins in $X(k_1,\ldots,k_n)(V)$, because $u_1$ is not adjacent to least one vertex $w\neq u_1$ in $O_{k_1}$ while $u_2$ is adjacent to this $w$. The same holds if reverse the roles of $v_1$ and $v_2$. Thus, we are left with the case when $v_1$ and $v_2$ are replaced by $K_{k_1}$ and $K_{k_2}$. In this case, any two vertices in $K_{k_1}\cup K_{k_2}$ are adjacent twins, and so the first statement in (2) holds. The second follows by using the same argument, and the third is a direct consequence of Theorem~\ref{sed}.
\end{proof}

Theorem~\ref{bup}(2) tells us that if $T$ is a set of non-adjacent (resp., adjacent) twins in $X$ and each vertex in $T$ is replaced by an empty (resp., complete) graph, one of which has size at least two, then $W=\bigcup_{v_j\in T} V(X_j)$ is a set of twins in $X(k_1,\ldots,k_n)(V)$ and each vertex in $W$ is $(1-\frac{2}{\lvert W\rvert})$-sedentary.

\begin{example}
Figure \ref{yay2} depicts three blow-ups of $C_4$: $C_4^2(V)$ is obtained by replacing each vertex of $C_4$ by $O_2$, $C_4(2,3,2,3)(V)$ by replacing two vertices of $C_4$ by two copies of $K_2$ and the rest by $O_3$, and $C_4^3(E)$ by replacing all edges of $C_4$ by copies of $O_3$. By Theorem~\ref{bup}(1), the vertices in the two copies of $O_3$ are $\frac{1}{3}$-sedentary in $C_4
(2,3,2,3)(V)$, while the 12 coloured vertices are $\frac{1}{3}$-sedentary in $C_4^3(E)$. By Theorem~\ref{bup}(2), a set of two twins in $C_4$ becomes a set of four in $C_4^2(V)$, all members of each set are $\frac{1}{2}$-sedentary.
\end{example}

\section{Cones}\label{secCones}

A graph of the form $K_1\vee X$ is called a \textit{cone} on $X$ with \textit{apex} $u$, where $V(K_1)=\{u\}$. A graph of the form $Z \vee X$, where $Z\in\{K_2,O_2\}$ is called a \textit{double cone} on $X$, and any vertex of $Z$ is called an \textit{apex}. In particular, $K_2\vee X$ and $O_2\vee X$ are resp. called \textit{connected} and \textit{disconnected} double cones.

For cones over $d$-regular graphs on $n$ vertices, Godsil showed that $\lvert U_A(t)_{u,u}\rvert\geq \frac{d^2}{d^2+4n}$, with equality if and only if $t=\frac{\pi}{\sqrt{d^2+4n}}$ \cite{SedQW}. This yields the following.

\begin{proposition}
\label{Godsil}
Let $d>0$ and $\mathscr{C}$ be a family of cones over weighted $d$-regular graphs on $n$ vertices.
\begin{enumerate}
\item If $d^2/n\rightarrow\infty$ as $n$ increases, then $\mathscr{C}$ is tightly sedentary at the apex.
\item If $\gamma$ is a constant such that $d^2/n\rightarrow \gamma$ as $n$ increases, then $\mathscr{F}$ is $\frac{\gamma}{\gamma+4}$-sedentary at the apex. In particular, if $d$ is fixed, then $\mathscr{C}$ is quasi-sedentary at the apex. 
\end{enumerate}
\end{proposition}

\begin{remark}
\label{conerem}
If $d=0$, then $\lvert U(\frac{\pi}{2\sqrt{n}})_{u,u}\rvert=0$, and so the apex in this case is not sedentary.
\end{remark}

\begin{theorem}
\label{adj}
For each $0\leq C\leq 1$, there exists a $C$-sedentary family with respect to the adjacency matrix.
\end{theorem}

\begin{proof}
If $0\leq C<1$, then $\mathscr{C}$ is $C$-sedentary at the apex by Proposition \ref{Godsil}(2) whenever $d^2/n\rightarrow \frac{4C}{1-C}$. If $d^2/n\rightarrow \infty$, then $\mathscr{C}$ is sedentary at the apex by Proposition \ref{Godsil}(1).
\end{proof}

For the Laplacian case, we prove a more general result for cones.

\begin{theorem}
\label{sedL1}
Let $m\geq 1$ and $X$ be a simple positively weighted graph on $n\geq 2$ vertices. For any vertex $u$ of $K_m$ in $K_m\vee X$, $\lvert U_L(t)_{u,u}\rvert\geq 1-\frac{2}{m+n}$ for all $t$ with equality if and only if $t=\frac{j\pi}{m+n}$ for some odd $j$. Thus, the family of joins $K_m\vee X$ is tightly sedentary at any vertex of $K_m$.
\end{theorem}
\begin{proof}
Let $u$ be a vertex of $K_m$ in $K_m\vee X$. By \cite[Equation 31]{Alvir2016}, $U_L(t)_{u,u}=\frac{1}{m+n}+\frac{m+n-1}{m+n}e^{it(m+n)}$. Thus, $\lvert U_L(t)_{u,u}\rvert^2\geq \frac{(m+n-2)^2}{(m+n)^2}$ for all $t$ and the result is immediate.
\end{proof}

By Corollary~\ref{sedg2}(1), if $m\geq 3$ is fixed, then $K_m\vee X$ is $(1-\frac{2}{m})$-sedentary at every vertex of $K_m$ with respect to $M$. But since $\lvert U_L(t)_{u,u}\rvert\geq 1-\frac{2}{m+n}>1-\frac{2}{m}$, this family of joins is, in fact, sedentary with respect to $L$. This also implies that Theorem~\ref{sedL1} yields a sharper bound than Theorem~\ref{sed}, which suggests that the bound obtained in Theorem~\ref{sed} can be improved if we take a more specific Hamiltonian. 

Taking $m\in\{1,2\}$ in Theorem~\ref{sedL1} yields the following result.

\begin{corollary}
\label{corsedLLL}
The families of cones and connected double cones on simple positively weighted graphs are tightly sedentary at the apexes with respect to the Laplacian matrix.
\end{corollary}

\section{Strongly cospectral vertices}\label{secSC}

We say that two vertices $u$ and $v$ are \textit{strongly cospectral} if $E_j\textbf{e}_u=\pm E_j\textbf{e}_v$ for all $\lambda_j\in\sigma_u(M)$. In this case,
\begin{center}
$\sigma_{uv}^+(M)=\{E_j\textbf{e}_u=E_j\textbf{e}_u\neq \textbf{0}\}\quad$ and $\quad \sigma_{uv}^-(M)=\{E_j\textbf{e}_u=-E_j\textbf{e}_u\neq \textbf{0}\}$.
\end{center}
partition $\sigma_u(M)$. The interest in the study of strongly cospectrality is motivated by the fact that it is a requirement for two vertices to exhibit PGST \cite[Lemma 13.1]{Godsil2012a}. In this section, we show that there are infinitely many graphs with strongly cospectral vertices that are sedentary. But as the next result shows, the machinery that we have developed in Section \ref{secSuff} has limitations for strongly cospectral vertices.

\begin{proposition}
\label{SCsed}
Let $u$ and $v$ are strongly cospectral. If $S=\sigma_{uv}^{\pm}(M)$, then $a=\frac{1}{2}$ in (\ref{eureka1}).
\end{proposition}

\begin{proof}
Assume $\sigma_{uv}^+(M)=\{\lambda_1,\ldots,\lambda_s\}$ and $\sigma_{uv}^-(M)=\{\lambda_{s+1},\ldots,\lambda_r\}$. Then we have $(E_{j})_{u,u}=(E_{j})_{u,v}$ for $j=1,\ldots,s$, while $(E_{j})_{u,u}=-(E_{j})_{u,v}$ for $j=s+1,\ldots,r$. As the $E_j$'s sum to identity, we get $\sum_{j=1}^s(E_{j})_{u,u}=\sum_{k=s+1}^r(E_{k})_{u,u}=\frac{1}{2}$.
\end{proof}

Let $u$ and $v$ be strongly cospectral. If we take $S\in\{\sigma_{uv}^+(M),\sigma_{uv}^-(M)\}$, then $a=\frac{1}{2}$ from Proposition \ref{SCsed}. In this case, Theorem~\ref{eureka} is not very useful, as (\ref{eureka2}) yields the trivial statement $\lvert U(t)_{u,u}\rvert\geq 0$ for all $t$. If we add that either (\ref{eureka3}) and (\ref{eureka333}) hold or Corollary~\ref{eurekarem2}(2) holds, then Corollary~\ref{singleton}(1b) implies that $u$ is not sedentary. Indeed, this holds because strong cospectrality together with either (\ref{eureka3}) and (\ref{eureka333}) or Corollary~\ref{eurekarem2}(2) resp. yield PST or PGST between $u$ and $v$. Hence, in order for Theorem~\ref{eureka} to work for strongly cospectral vertices, one may avoid taking $S\in\{\sigma_{uv}^+(M),\sigma_{uv}^-(M)\}$. For the case of strongly cospectral twin vertices, it is known that $\lvert\sigma_{uv}^-(M)\rvert=1$ \cite[Theorem 3.4]{MonterdeELA}), and so the only viable option is to choose $S$ such that $\sigma_{uv}^-(M)$ is a proper subset of $S$, in which case, $\lvert S\rvert\geq 2$ and $\theta\in S$, where $\theta$ is given in (\ref{adjalpha}). However, we shall see in Remark~\ref{twoeval} of the next subsection that, unlike the case $S=\{\theta\}$ which yields Theorem~\ref{sed}, the case $\lvert S\rvert\geq 2$ with $\theta\in S$ requires more work in order to establish sedentariness of $u$.

To achieve our goal of showing that there are infinitely many graphs with strongly cospectral vertices that are sedentary, we consider disconnected double cones. Indeed, the apexes of such graphs are strongly cospectral with respect to $A$ and $L$ by \cite[Corollary 6.9]{MonterdeELA}. Our main motivation for considering these graphs is that their apexes form a set of twins of size two, and our results in Corollary~\ref{corjoin1}(1) prompt us to investigate whether the apexes of $O_2\vee X$ are also sedentary. Results in the literature indicate that the apexes of disconnected double cones are excellent sources of PST and PGST (see for instance \cite{Alvir2016} for the Laplacian case and \cite{Monterde2022} for the adjacency and signless Laplacian case), and so one might be inclined to speculate that these apexes are not sedentary. But as it turns out, the apexes of disconnected double cones are sedentary whenever they do not exhibit PST or PGST.

\subsection*{Laplacian case}

Let $X$ be a simple positively weighted graph on $n$ vertices. Then
\begin{equation}
\label{om}
U_L(t)_{u,u}=\frac{1}{m+n}+\frac{(m-1)e^{itn}}{m}+\frac{ne^{it(m+n)}}{m(m+n)}
\end{equation}
for each vertex $u$ of $O_m$ in $O_m\vee X$ (see \cite[Equation 33]{Alvir2016}).

\begin{theorem}
\label{sedL}
Let $X$ be a simple positively weighted graph on $n$ vertices, and let $u$ be an apex of $O_2\vee X$.
\begin{enumerate}
\item If $n\equiv 2$ (mod 4), then $O_2\vee X$ is not sedentary at $u$.
\item If $n\equiv 0$ (mod 4), then $\lvert U_L(t)_{u,u}\rvert\geq 1-\frac{n}{n+2}$ with equality if and only if $t=\frac{j\pi}{2}$ for any odd $j$.
\item Let $n$ be odd. If $n=1$, then $\lvert U_L(t)_{u,u}\rvert \geq \frac{1}{3}$ with equality if and only if $t=\ell\pi$ for any odd $\ell$, while $n\geq 3$, then $\lvert U_L(t)_{u,u}\rvert \geq 1-\frac{n+2-\sqrt{2}}{n+2}$ with equality if and only if $t=\frac{j\pi}{2}$ for any odd $j$.
\end{enumerate}
\end{theorem}

\begin{proof}
Let $u$ be an apex of $O_2\vee X$. From (\ref{om}), $U_L(t)_{u,u}=\frac{1}{n+2}+\frac{e^{itn}}{2}+\frac{ne^{it(n+2)}}{2(n+2)}$, and so
\begin{equation}
\begin{split}
\label{dc}
\lvert U_L(t)_{u,u}\rvert^2&=\frac{n^2+2n+4+h(t)}{2(n+2)^2},
\end{split}
\end{equation}
where $h(t)=n(n+2)\cos(2t)+2(n+2)\cos(tn)+2n\cos(t(n+2))$. Note that $\lvert U_L(t)_{u,u}\rvert^2$ is maximum (resp., minimum) if and only if $h(t)$ is maximum (resp., minimum). One can then verify that
\begin{center}
$\hspace{-0.152in} h'(t)\stackrel{(*)}{=}-2n(n+2)\left[\sin(2t)+\sin(tn)+\sin(t(n+2))\right]=-8n(n+2)\cos\left(t\right)\cos\left(tn/2\right)\sin\left(t(n+2)/2\right)$.
\end{center}
Thus, $h'(t)=0$ if and only if either (i) $t=j\pi/2$ for some odd $j$, (ii) $t=\ell\pi/n$ for some odd $\ell$, or (iii) $t=\frac{k\pi}{n+2}$ for some even $k$. We now differentiate $h'(t)$ in $(*)$ to get
\begin{equation}
\label{pp}
h''(t)=-4n(n+2)\left[2\cos\left(t(n+4)/2\right)\cos\left(tn/2\right)+n\cos\left(t(n+1)\right)\cos(t)\right]
\end{equation}
If $t=\frac{j\pi}{2}$ for some odd $j$, then $\cos(t)=0$ and $\cos\left(t(n+4)/2\right)=-\cos\left(tn/2\right)$ because $j$ is odd. While if $t=\ell\pi/n$ for some odd $\ell$, then $\cos\left(tn/2\right)=0$ and $\cos(t(n+1))=-\cos(t)$ because $\ell$ is odd. In both cases, (\ref{pp}) yields $h''(t)>0$, and so $\lvert U_L(t)_{u,u}\rvert^2$ has a relative minimum. Now, if $t=\frac{k\pi}{n+2}$ for some even $k$, then $\cos\left(t(n+4)/2\right)=\cos\left(tn/2\right)$ and $\cos\left(t(n+1)\right)=\cos(t)$. Using (\ref{pp}), one checks that $h''(t)<0$, and so $\lvert U_L(t)_{u,u}\rvert^2$ has a relative maximum. From these three cases, it suffices to compare the values of $\lvert U_L(t)_{u,u}\rvert^2$ at $t=j\pi/2$ for odd $j$ and $t=\ell\pi/n$ for odd $\ell$ to get the absolute minimum. We begin with $t=j\pi/2$ for some odd $j$. In this case, $\cos(2t)=-1$ and $\cos(t(n+2))=-\cos(tn)$, and so (\ref{dc}) yields
\begin{equation}
\label{cos1}
\begin{split}
\lvert U_L(t)_{u,u}\rvert^2=\frac{4\left[1+\cos\left(tn\right)\right]}{2(n+2)^2}.
\end{split}
\end{equation}
If $n\equiv 2$ (mod 4), then $u$ exhibits PST \cite[Corollary 4]{Alvir2016}, and so it is not sedentary. This proves (a). Thus, we have two remaining cases.
\begin{itemize}
\item If $n=4m$, then $\cos(tn)=\cos\left(2jm\pi\right)=1$, and so (\ref{cos1}) yields $\lvert U_L(t)_{u,u}\rvert^2=\frac{4}{(n+2)^2}$.
\item If $n$ is odd, then $\cos(tn)=0$, and so (\ref{cos1}) gives us $\lvert U_L(t)_{u,u}\rvert^2=\frac{2}{(n+2)^2}$.
\end{itemize}
Next, let $t=\frac{\ell\pi}{n}$ for odd $\ell$. Then $\cos(tn)=-1$ and $\cos(t(n+2))=-\cos(2t)$. From (\ref{dc}), we obtain
\begin{equation}
\label{cos}
\lvert U_L(t)_{u,u}\rvert^2=\frac{n^2(1+\cos(2t))}{2(n+2)^2}.
\end{equation}
\begin{itemize}
\item Let $n=4m$. Then $2t=\frac{\ell\pi}{2m}$ cannot be an odd multiple of $\pi$, and so $\cos(2t)>-1$. The closest that $2t$ will be from an odd multiple of $\pi$ is when $\ell=2ms\pm 1$ for some odd $s$, in which case, $2t=\left(s\pm \frac{1}{2m}\right)\pi$. From (\ref{cos}), we get $\lvert U_L(t)_{u,u}\rvert^2\geq \frac{n^2\left(1-\cos\left(2\pi/n\right)\right)}{2(n+2)^2}=\frac{n^2\sin^2\left(\pi/n\right)}{(n+2)^2}\geq \frac{8}{(n+2)^2}$.
\item Let $n$ be odd. If $n=1$, then $2t=2\ell\pi$ and so $\lvert U_L(t)_{u,u}\rvert^2=\frac{1}{9}$. If $n>1$, then using (\ref{cos}) and the same argument in the case $n=4m$ yields $\lvert U_L(t)_{u,u}\rvert^2\geq \frac{n^2\left(1-\cos\left(\pi/n\right)\right)}{2(n+2)^2}=\frac{n^2\sin^2\left(\pi/2n\right)}{(n+2)^2}\geq \frac{2.25}{(n+2)^2}$.
\end{itemize}
Finally, comparing the above subcases yields the desired result.
\end{proof}

\begin{remark}
\label{twoeval}
In the above proof, if we take $S=\{0,n\}$, then $\sigma_u(L)\backslash S=\{n+2\}$ and $(E_n)_{u,u}+(E_{0})_{u,u}=a$, where $a=1-\frac{n}{2(n+2)}>\frac{1}{2}$. If $F(t):=\lvert \frac{1}{n+2}+\frac{1}{2}e^{itn}\rvert-\frac{n}{2(n+2)}$, then Theorem~\ref{eureka} yields $\vert U(t)_{u,u}\rvert \geq F(t)\geq 0$ for all $t$. If $n\equiv 0$ (mod 4), then one checks that (\ref{eureka333}) holds at $t_1=j\pi/2$ for odd $j$, i.e., the first inequality is an equality, in which case $U(t_1)=2a-1=1-\frac{n}{n+2}$. As $F(t)$ is not a constant function, we are not guaranteed that $\lvert U(t)_{u,u}\rvert\geq 1-\frac{n}{n+2}$ for all $t$. To show that this is indeed the case, we need to establish that $\lvert U(t)_{u,u}\rvert$ is minimum at $t_1$, which was the approach taken in the proof of Theorem~\ref{sedL}.
\end{remark}

The following is an immediate consequence of Theorem~\ref{sedL}.

\begin{corollary}
\label{corsedL}
The family of disconnected double cones on simple positively weighted graphs on $n$ vertices, where $n\not\equiv 2$ (mod 4), is quasi-sedentary at the apexes with respect to the Laplacian matrix. 
\end{corollary}

In keeping with Theorem~\ref{sedL1}, we show that the bound in Theorem~\ref{sed} is tight for the vertices of $O_m$ in $O_m\vee X$ for some values of $m$ and $n$.

\begin{theorem}
\label{sedom}
Let $m\geq 3$ and $X$ be a simple positively weighted graph on $n\geq 1$ vertices. If $\nu_2(m)=\nu_2(n)$, then for any vertex $u$ of $O_m$ in $O_m\vee X$, $\lvert U_L(t)_{u,u}\rvert=1-\frac{2}{m}$ whenever $t_1=\frac{j\pi}{\operatorname{gcd}(m,n)}$ for any odd $j$.
\end{theorem}

\begin{proof}
By Theorem~\ref{sed}, $\lvert U_L(t)_{u,u}\rvert\geq 1-\frac{2}{m}$ for all $t$. Letting $t_1=\frac{j\pi}{g}$, where $j$ is odd and $g=\operatorname{gcd}(m,n)$, one can check that using (\ref{om}) that $\lvert U_L(t)_{u,u}\rvert= 1-\frac{2}{m}$.
\end{proof}

Let $m$ be fixed. By Theorem~\ref{sedom}, the family of joins $O_m\vee X$ such that $\nu_2(m)=\nu_2(n)$ is tightly $(1-\frac{2}{m})$-sedentary at the vertices of $O_m$. Our numerical observations indicate that if $\nu_2(m)\neq \nu_2(n)$ with $m$ fixed, then any vertex of $O_m$ is $C(n)$-sedentary in $O_m\vee X$, where $C(n)$ satisfies $C(n)<1-\frac{2}{m}$ for all $n$ and $C(n)\rightarrow 1-\frac{2}{m}$ as $n$ increases. This suggests that in general, the family of graphs of the form $O_m\vee X$ with fixed $m$ is sharply $(1-\frac{2}{m})$-sedentary at the vertices of $O_m$.

\subsection*{Adjacency case}

Next, we examine the adjacency case for disconnected double cones. 

\begin{theorem}
\label{sedA}
Let $X$ be a simple uweighted $d$-regular graph on $n$ vertices, and let $u$ be an apex of $O_2\vee X$.
\begin{enumerate}
\item If either (i) $d^2+8n$ is not a perfect square, (ii) $d=0$, or (iii) $d^2+8n$ is a perfect square and $\nu_2(d+\sqrt{d^2+8n})=\nu_2(d-\sqrt{d^2+8n})$, then $O_2\vee X$ is not sedentary at $u$.
\item Let $d>0$ and $n=\frac{1}{2}s(d+s)$ for some integer $s>0$ such that $s(d+s)$ is even and $\nu_2(d+s)\neq \nu_2(s)$ (i.e., $d^2+8n$ is a perfect square and $\nu_2(d+\sqrt{d^2+8n})\neq \nu_2(d-\sqrt{d^2+8n})$). Suppose $d_1=d/g$ and $s_1=s/g$, where $g=\operatorname{gcd}(d,s)$. The following hold.
\begin{enumerate}
\item If $s_1=1$, then $\lvert U_A(t)_{u,u}\rvert\geq \frac{1}{d_1+2}$ with equality if and only if $t=\frac{j\pi}{g}$ for any odd $j$.
\item If $s_1\geq 2$, then $\lvert U_A(t)_{u,u}\rvert\geq \frac{\sqrt{2}}{d_1+2s_1}$ with equality if and only if $t=\frac{j\pi}{g}$ for any odd $j$.
\end{enumerate}
\end{enumerate}
\end{theorem}

\begin{proof}
Let $Y$ be a double cone on $X$ with apexes $u$ and $v$. From \cite[Lemma 3(2)]{Monterde2022}, we know that $\sigma_u(A)=\{0,\lambda^{\pm}\}$, where $\lambda^{\pm}=\frac{1}{2}\left(d\pm\sqrt{d^2+8n}\right)$. Applying \cite[Lemma 12.3.1]{Coutinho2021}, we obtain
\begin{equation}
\label{UA}
U_A(t)_{u,u}=\frac{1}{2}+\frac{n}{2n+(\lambda^+)^2}e^{it\lambda^+}+\frac{n}{2n+(\lambda^-)^2}e^{it\lambda^-}
\end{equation}
If $d^2+8n$ is not a perfect square, then PGST occurs between $u$ and $v$ \cite[Theorem 11(1)]{Monterde2022}, while if $d=0$, or $d^2+8n$ is a perfect square and $\nu_2(d+\sqrt{d^2+8n})=\nu_2(d-\sqrt{d^2+8n})$, then PST occurs between $u$ and $v$ \cite[Theorem 14(2a)]{Monterde2022}. Invoking Proposition \ref{prop}(1c) yields (1a). To prove (1b), suppose $d^2+8n$ is a perfect square, $d>0$, and $\nu_2(d+\sqrt{d^2+8n})\neq \nu_2(d-\sqrt{d^2+8n})$. This is equivalent to the existence of an integer $s>0$ such that $n=\frac{1}{2}s(d+s)$ with $s(d+s)$ is even and
\begin{equation}
\label{nu2}
\nu_2(s)\geq \nu_2(d). 
\end{equation}
One also checks that $2n+(\lambda^+)^2=(d+s)(d+2s)$ and $2n+(\lambda^-)^2=s(d+2s)$. Thus, we obtain $\frac{n}{2n+(\lambda^+)^2}=\frac{s(d+s)}{2(d+s)(d+2s)}=\frac{s}{2(d+2s)}$ and $\frac{n}{2n+(\lambda^-)^2}=\frac{s(d+s)}{2s(d+2s)}=\frac{d+s}{2(d+2s)}$. Combining this with (\ref{UA}), we get that $U_A(t)_{u,u}=\frac{1}{2(d+2s)}\left((d+2s)+se^{it(d+s)}+(d+s)e^{-its}\right)$, and so
\begin{equation}
\label{dcA}
\lvert U_A(t)_{u,u}\rvert^2=\frac{d^2+3ds+3s^2+h(t)}{2(d+2s)^2}
\end{equation}
where $h(t)=s(d+s)\cos(t(d+2s))+s(d+2s)\cos(t(d+s))+(d+s)(d+2s)\cos(ts)$. Following the Laplacian case, $\lvert U_A(t)_{u,u}\rvert^2$ is maximum (resp., minimum) if and only if $h(t)$ is, and we have
\begin{equation}
\label{g'}
\begin{split}
h'(t)&=-4s(d+s)(d+2s)\cos\left(ts/2\right)\cos\left(t(d+s)/2\right)\sin\left(t(d+2s)/2\right).
\end{split}
\end{equation}
and
\begin{equation}
\label{g''}
h''(t)=-2s(d+s)(d+2s)\left[(d+s)\cos\left(t(2d+3s)/2\right)\cos\left(ts/2\right)+s\cos\left(t(d+3s)/2\right)\cos\left(t(d+s)/2\right)\right].
\end{equation}
From (\ref{g'}), $h'(t)=0$ if and only if either $t=j\pi/s$ for some odd $j$, $t=\frac{\ell\pi}{d+s}$ for some odd $\ell$, and $t=\frac{k\pi}{d+2s}$ for even $k\neq 0$. Among these three values, one can use (\ref{g''}) to show that $h''(t)>0$ if and only if $t=j\pi/s$ for odd $j$ and $t=\frac{\ell\pi}{d+s}$ for odd $\ell$. Thus, it suffices to compare the values of $\lvert U_A(t)_{u,u}\rvert^2$ at the points $t=j\pi/s$ for some odd $j$ and $t=\frac{\ell\pi}{d+s}$ for some odd $\ell$ to obtain the absolute minimum. Let's start with $t=j\pi/s$ for some odd $j$. In this case, $\cos(ts)=-1$, $\cos(t(d+2s))=\cos\left(jd\pi/s\right)$ and $\cos(t(d+s))=-\cos\left(jd\pi/s\right)$. Thus, $h(t)=-s^2\cos\left(jd\pi/s\right)-(d+s)(d+2s)$, and (\ref{dcA}) gives us
\begin{equation}
\label{sin}
\lvert U_A(t)_{u,u}\rvert^2=\frac{s^2\left[1-\cos\left(jd\pi/s\right)\right]}{2(d+2s)^2}
\end{equation}
Let $d=gd_1$ and $s=gs_1$, where $g=\operatorname{gcd}(d,s)$. Then we can write $jd\pi/s=jd_1\pi/s_1$, where $d_1$ is odd by (\ref{nu2}). We proceed with two subcases.
\begin{itemize}
\item Let $s_1=1$. Then (\ref{sin}) yields $\lvert U_A(t)_{u,u}\rvert^2=\frac{1-\cos\left(jd_1\pi\right)}{2(d_1+2)^2}=\frac{1}{(d_1+2)^2}$.
\item Let $s_1\geq 2$. As $j$ is odd, $jd_1\pi/s_1$ is not an even multiple of $\pi$. Hence, $\cos\left(jd_1\pi/s_1\right)\leq \cos\left(\pi/s_1\right)$, and making use of (\ref{sin}) then yields
$\lvert U_A(t)_{u,u}\rvert^2\geq \frac{s_1^2\left[1-\cos\left(\pi/s_1\right)\right]}{2(d_1+2s_1)^2}=\frac{s_1^2\sin^2\left(\pi/2s_1\right)}{(d_1+2s_1)^2}\geq \frac{2}{(d_1+2s_1)^2}$.
\end{itemize}
Next, let $t=\frac{\ell\pi}{d+s}$ for some odd $\ell$. In this case, $\cos(t(d+s))=-1$ and $\cos(t(d+2s))=\cos\left(\frac{\ell s\pi}{d+s}\right)$. Thus, $h(t)=-s(d+2s)+(d+s)^2\cos(ts)$, and making use of (\ref{dcA}) gives us
\begin{equation}
\label{sin1}
\lvert U_A(t)_{u,u}\rvert^2=\frac{(d+s)^2\left[1+\cos\left(\frac{\ell s\pi}{d+s}\right)\right]}{2(d+2s)^2}
\end{equation}
Note that we can write $\frac{\ell s\pi}{d+s}=\frac{\ell s_1\pi}{d_1+s_1}$. We proceed with two subcases.
\begin{itemize}
\item Let $s_1=1$. The same argument in the proof of Theorem~\ref{sedL} for the case $t=\frac{\ell\pi}{n}$ for odd $\ell$ and $n=4m$ yields $\lvert U_A(t)_{u,u}\rvert^2\geq \frac{(d_1+1)^2\left[1-\cos\left(\frac{\pi}{d_1+1}\right)\right]}{2(d_1+2)^2}=\frac{(d_1+1)^2\sin^2\left(\frac{\pi}{2(d_1+1)}\right)}{(d_1+2)^2}\geq \frac{2}{(d_1+2)^2}$.
\item Let $s_1\geq 2$. Since $d_1+s_1$ and $s_1$ must have opposite parities, $\frac{\ell s_1\pi}{d_1+s_1}$ is not an odd multiple of $2\pi$. Thus, (\ref{sin1}) yields $\lvert U_A(t)_{u,u}\rvert^2\geq \frac{(d_1+s_1)^2\left[1-\cos\left(\frac{\pi}{d_1+s_1}\right)\right]}{2(d_1+2s_1)^2}=\frac{(d_1+s_1)\sin^2\left(\frac{\pi}{2(d_1+s_1)}\right)}{(d_1+2s_1)^2}\geq \frac{2.25}{(d_1+2s_1)^2}$.
\end{itemize}
Comparing the minima for each subcases above yields the desired conclusion.
\end{proof}

The following is a straightforward consequence of Theorem~\ref{sedA}(2).

\begin{corollary}
\label{corsedA}
With the assumption in Theorem~\ref{sedA}(2), let $\mathscr{F}$ be a family of disconnected double cones on simple unweighted $d$-regular graphs on $n$ vertices. If $s_1$ and $d_1$ are fixed, then $\mathscr{F}$ is $\frac{1}{d_1+2}$- and $\frac{\sqrt{2}}{d_1+2s_1}$-sedentary at the apexes resp.\ whenever $s_1=1$ and $s_1\geq 2$. On the other hand, if either $d_1\rightarrow\infty$ or $s_1\rightarrow\infty$ as $n$ increases, then $\mathscr{F}$ is quasi-sedentary at the apexes.
\end{corollary}

We end this section with the following remark.

\begin{remark}
Let $u$ and $v$ be the apexes of $O_2\vee X$, where $X$ is regular whenever $M=A$.
\begin{enumerate}
\item Suppose the assumption Theorem~\ref{sedL} holds. If $n\equiv 0$ (mod 4), then $\lvert U_L(t)_{u,v}\rvert \leq 1-\delta$ for all $t$, where $\delta=\frac{2}{n+2}$, while if $n\geq 3$ is odd, then $\lvert  U_L(t)_{u,v}\rvert \leq 1-\delta$ for all $t$, where $\delta=\frac{\sqrt{2}}{n+2}$.
\item Suppose the assumption in Theorem~\ref{sedA}(2) holds. If $s_1=1$, then $\lvert U_A(t)_{u,v}\rvert \leq 1-\delta$ for all $t$, where $\delta=\frac{1}{d_1+2}$, while if $s_1\geq 2$, then $\rvert U_A(t)_{u,v}\rvert \leq 1-\delta$ for all $t$, where $\delta=\frac{\sqrt{2}}{d_1+2s_1}$.
\end{enumerate}
In \cite{Godsil2017}, Godsil and Smith asked to find examples of strongly cospectral vertices $u$ and $v$ such that for some constant $\delta>0$, $\lvert U(t)_{u,v}\rvert \leq 1-\delta$ for all $t$. Mirror symmetric vertices in paths without PGST and antipodal vertices in even cycles without PGST are infinite families that answer this question. However, we do not know whether paths and cycles are sedentary. Thus, the families in (1) and (2) are the first examples that answer Godsil and Smith's question, whereby the vertices involved are sedentary.
\end{remark}

\section{Trees}\label{secTrees}

Our first result in this section is a direct consequence of Theorem~\ref{sed}.

\begin{proposition}
\label{treeprop}
Let $T$ be a set of leaves of a tree $X$ that share a common neighbour. Then $T$ is a set of twins in $X$, and for each $u\in T$, we have $\lvert U_M(t)_{u,u}\rvert\geq 1-\frac{2}{\lvert T\rvert}$ for all $t$.
\end{proposition}

Next, we examine whether the central vertex of a star is sedentary.

\begin{proposition}
\label{starsed}
Let $T$ be the set of leaves of $K_{1,n}$ and $u\in T$. Then $\lvert U_M(t)_{u,u}\rvert\geq 1-\frac{2}{n}$ for all $t$. Hence, the family $\mathscr{S}$ of stars $K_{1,n}$ on $n\geq 3$ vertices is sedentary at every leaf vertex. The following also hold.
\begin{enumerate}
\item For all $n\geq 2$, $\lvert U_A(t)_{u,u}\rvert= 1-\frac{2}{n}$ if and only if $t=\frac{j\pi}{\sqrt{n}}$ for any odd $j$. Moreover, the central vertex of $K_{1,n}$ is not sedentary with respect to $A$.
\item If $n$ is odd, then $\lvert U_L(t)_{u,u}\rvert= 1-\frac{2}{n}$ whenever $t=j\pi$ for any odd $j$. For all $n\geq 2$, the central vertex $w$ of $K_{1,n}$ satisfies $\lvert U_L(t)_{w,w}\rvert\geq 1-\frac{2}{n+1}$ with equality if and only if $t=\frac{j\pi}{n+1}$ for any odd $j$.
\end{enumerate}
\end{proposition}

\begin{proof}
The first statement follows from Proposition \ref{treeprop}, while the second one is obtained by by noting that $\lvert V(X)\backslash T\rvert=1$ and applying Corollary~\ref{sedF}(2). To prove (1a), one can use the fact that $\sigma_u(A)=\{0,\pm\sqrt{n}\}$ to show that (\ref{eureka3}) and (\ref{eureka333}) hold if and only if $t_1=\frac{j\pi}{\sqrt{n}}$ for odd $j$. Thus, equality holds in (\ref{UHH}) in Theorem~\ref{sed}, which yields the first statement of (1a). As $K_{1,n}$ is a cone on a $0$-regular graph, the second statement follows from Remark~\ref{conerem}. Finally, since $K_{1,n}=O_n\vee K_1$ with $T=V(O_n)$, Theorem~\ref{sedom} yields the first statement of (2), while the second follows by noting that $K_{1,n}=O_n\vee K_1$ and applying Theorem~\ref{sedL1}.
\end{proof}

Proposition \ref{starsed}(2) implies that the apex of a cone on any simple positively weighted graph $X$ on $n$ vertices is Laplacian tightly $(1-\frac{2}{n+1})$- sedentary. Next, we use Proposition \ref{starsed} to create more sedentary families using Cartesian products.
 
\begin{example}
\label{stars}
Let $k\geq 1$, $n\geq 3$ and $Z_{k,n}$ be the Cartesian product of $K_{1,n}$ with itself $k$ times. Let $u$ be a leaf of $K_{1,n}$. By Proposition \ref{starsed}, $\lvert U_{Z_{k,n}}(t)_{(u,\ldots,u),(u,\ldots,u)}\rvert\geq (1-\frac{2}{n})^k$ for all $t$ with respect to $M$. Consider the families $\mathscr{F}_1=\{Z_{k,n}:k\ \text{fixed}\}$, $\mathscr{F}_2=\{Z_{k,n}:n=\lfloor mk\rfloor \ \text{for some}\ m>0\}$ and $\mathscr{F}_3=\{Z_{k,n}:n\ \text{fixed}\}$.
\begin{enumerate}
\item As $n$ increases, $(1-\frac{2}{n})^k\rightarrow 1$ in $\mathscr{F}_1$ and $(1-\frac{2}{n})^k\rightarrow 1/\sqrt[m]{e^2}$ in $\mathscr{F}_2$. Thus, $\mathscr{F}_1$ is sedentary while $\mathscr{F}_2$ is $1/\sqrt[m]{e^2}$-sedentary at $(u,\ldots,u)$. If $M=A$, then the sedentariness of $\mathscr{F}_1$ and $\mathscr{F}_2$ is tight by Proposition \ref{starsed}(1). If $M=L$, then the sedentariness of the subfamilies of $\mathscr{F}_1$ and $\mathscr{F}_2$ is tight whenever $n$ in each $Z_{k,n}$ is odd by virtue of Proposition \ref{starsed}(2).
\item Since $(1-\frac{2}{n})^k\rightarrow 0$ as $k\rightarrow\infty$, $\mathscr{F}_3$ is quasi-sedentary by Proposition \ref{starsed}.
\end{enumerate}
\end{example}

If $u$ in Example~\ref{stars} is instead the degree $n$ vertex of $K_{1,n}$, then $\lvert U_{Z_{k,n}}(t)_{(u,\ldots,u),(u,\ldots,u)}\rvert \geq (1-\frac{2}{n+1})^k$ for all $t$ with respect to $L$ by Proposition \ref{starsed}(2). Thus, $\mathscr{F}_1$ is sedentary, $\mathscr{F}_2$ is $1/\sqrt[m]{e^2}$-sedentary, and $\mathscr{F}_3$ is quasi-sedentary at $(u,\ldots,u)$ with respect to $L$. Now, this does not hold for $A$ by Proposition \ref{starsed}, and so we get a family that is sedentary with respect to $L$ but not to $A$. By Theorems \ref{sedL} and \ref{sedA}, one may construct a family that is sedentary with respect to $A$ but not to $L$ by taking the family of disconnected double cones on $d$-regular graphs on $n\equiv 2$ (mod 4) vertices satisfying condition (2) of Theorem~\ref{sedA}.

We also note that if we replace the $Z_{k,n}$'s in the above example by the Hamming graphs $H(k,n)$, then $\mathscr{F}_1$ is tightly sedentary (which we already know by Corollary~\ref{ham}), $\mathscr{F}_2$ is tightly $1/\sqrt[m]{e^2}$-sedentary at $(u,\ldots,u)$, and $\mathscr{F}_3$ is quasi-sedentary. Moreover, this applies to any $M$ because each $H(k,n)$ is regular.

A \textit{double star} $S_{k,\ell}$ is a tree resulting from attaching $k$ and $\ell$ pendent vertices to the vertices of $K_2$. Like the central vertex of $K_{1,n}$, we show that an internal vertex of $S_{k,k}$ is not sedentary whenever $M=A$.

\begin{theorem}
\label{doublestars}
Let $k\geq 1$, and consider a double star $S_{k,k}$ with internal vertices $u$ and $v$.
\begin{enumerate}
\item Let $4k+1$ be a perfect square, and let $w$ be a leaf of $S_{k,k}$. If $k=2$, then $\lvert U_A(t)_{w,w}\rvert\geq \frac{1}{4}$, with equality whenever $t=\frac{j\pi}{3}$, where $j\equiv 2,4$ (mod 6). If $k>2$, then $\lvert U_A(t)_{w,w}\rvert\geq 1-\frac{2}{k}$, with equality whenever $t=j\pi$ for an integer $j$ such that $j\sqrt{4k+1}\equiv 3$ (mod 4).
\item For all $k\geq 1$, $u$ and $v$ are not sedentary in $S_{k,k}$ with respect to the adjacency matrix. 
\end{enumerate}
\end{theorem}

\begin{proof}
Suppose we index the vertices of $A(S_{k,\ell})$ starting with the $k$ leaves attached to $u$, followed by $u$ and $v$, and then the $k$ leaves attached to $v$. Then we can write $A(S_{k,k})=\left[\begin{array}{cccc} A(K_{1,k})&Y\\ Y^T&A(K_{1,k})\end{array} \right]$, where $A(K_{1,k})=\left[\begin{array}{cccc} \textbf{0}_{k}&\textbf{1} \\ \textbf{1}^T&0 \end{array} \right]$ and $Y=\left[\begin{array}{cccc} \textbf{0}&\textbf{0}_k \\ 1&\textbf{0}\end{array} \right]$. Thus, $\textbf{e}_1-\textbf{e}_j$ for $j=1,\ldots,k$ and $\textbf{e}_{k+2}-\textbf{e}_j$ for $j=k+3,\ldots,2k+2$ are eigenvectors for $A(S_{k,k})$ corresponding to the eigenvalue 0. Moreover,
\begin{equation}
\label{evals}
\hspace{-0.15in} \lambda_1=-\frac{1}{2}(1+\sqrt{4k+1}), \lambda_2=\frac{1}{2}(-1+\sqrt{4k+1}), \lambda_3=\frac{1}{2}(1-\sqrt{4k+1})\ \ \text{and}\ \ \lambda_4=\frac{1}{2}(1+\sqrt{4k+1})
\end{equation}
are simple eigenvalues of $A(S_{k,k})$ resp.\ with eigenvectors $\textbf{v}_1=\left[-\textbf{1},\frac{1+\sqrt{4k+1}}{2},-\frac{1+\sqrt{4k+1}}{2},\textbf{1}\right]$, $\textbf{v}_2=\left[-\textbf{1},\frac{1-\sqrt{4k+1}}{2},\frac{-1+\sqrt{4k+1}}{2},\textbf{1}\right]$, $\textbf{v}_3=\left[\textbf{1},\frac{1-\sqrt{4k+1}}{2},\frac{1-\sqrt{4k+1}}{2},\textbf{1}\right]$ and $\textbf{v}_4=\left[\textbf{1},\frac{1+\sqrt{4k+1}}{2},\frac{1+\sqrt{4k+1}}{2},\textbf{1}\right]$. Thus, $E_0= I_k-\frac{1}{k}\textbf{J}_k\oplus O_2\oplus  I_k-\frac{1}{k}\textbf{J}_k$, where $A\oplus B$ is the direct sum of matrices $A$ and $B$, and  $E_{\lambda_j}=\frac{1}{\|\textbf{v}_j\|^2}\textbf{v}_j^T\textbf{v}_j$ for $j\in\{1,2,3,4\}$. Thus, $\sigma_u(A)=\{\lambda_1,\lambda_2,\lambda_3,\lambda_4\}$ and $\sigma_w(A)=\{0\}\cup \sigma_u(A)$ for any leaf $w$ of $S_{k,k}$. Using spectral decomposition and the fact that $\lambda_1=-\lambda_4$ and $\lambda_2=-\lambda_3$ yields
\begin{equation}
\label{doubs}
\begin{split}
U_A(t)_{u,u}&=\frac{(1+\sqrt{4k+1})^2\cos(t\lambda_1)}{2(4k+1+\sqrt{4k+1})}+\frac{(1-\sqrt{4k+1})^2\cos(t\lambda_2)}{2(4k+1-\sqrt{4k+1})}
\end{split}
\end{equation}
and
\begin{equation}
\label{Ustar1}
\begin{split}
U_A(t)_{w,w}&=\frac{k-1}{k}+\frac{2\cos(t\lambda_1)}{4k+1+\sqrt{4k+1}}+\frac{2\cos(t\lambda_2)}{4k+1-\sqrt{4k+1}}.
\end{split}
\end{equation}
Let $4k+1$ be a perfect square. Note that the first statement of (1) can be easily verified using (\ref{Ustar1}). Since Theorem~\ref{sed} yields $\lvert U_A(t)_{w,w}\rvert \geq 1-\frac{2}{k}$ for all $t$, one can check using (\ref{Ustar1}) that indeed, $\lvert U_A(t_1)_{w,w}\rvert=1-\frac{2}{k}$ whenever $t_1=j\pi$, where $j$ is an integer such that $j\ell\equiv 3$ (mod 4). This proves (1). To prove (2), it suffices to check the case when $4k+1$ is a perfect square by Proposition \ref{prop}(1c) and \cite[Theorem 5.3]{Fan2013}. Observe that $U_A(t)_{u,u}$ in (\ref{doubs}) is a real valued continuous function that has positive and negative values as $t$ ranges across $[0,2\pi]$. By IVT, there exists a $t\in [0,2\pi]$ such that $U_A(t)_{u,u}=0$, i.e., $u$ is not sedentary.
\end{proof}

\begin{corollary}
\label{doublestars1}
Let $u$ be a vertex of $S_{k,\ell}$ with $\operatorname{deg}(u)=k+1$. The following hold for $M=A$.
\begin{enumerate}
\item If $k=2$, then the leaves attached to $u$ in $S_{2,\ell}$ are sedentary if and only if $\ell=2$. In particular, if $k=\ell=2$, then all leaves in $S_{2,2}$ are tightly $\frac{1}{4}$-sedentary. 
\item If $k\geq 3$ and the nonzero eigenvalues of $A(S_{k,\ell})$ are linearly independent over $\mathbb{Q}$, then the leaves attached to $u$ are sharply $(1-\frac{2}{k})$-sedentary. In particular, if $4k+1$ is not a perfect square, then all leaves of $S_{k,k}$ are sharply $(1-\frac{2}{k})$-sedentary.
\end{enumerate}
\end{corollary}

\begin{proof}
If $k=2$, then PGST occurs between the leaves attached to $u$ if and only if $\ell=2$ \cite[Theorem 5.2]{Fan2013}. By Proposition \ref{prop}(1c), we only need to check $S_{2,2}$. By Theorem~\ref{doublestars}(1), a leaf attached to $u$ is tightly $\frac{1}{4}$-sedentary. As all leaves in $S_{2,2}$ are cospectral, Proposition \ref{prop}(1b) implies they are tightly $\frac{1}{4}$-sedentary. This proves (1). Now, assume the premise of (2). Take $S=\{0\}$ in Theorem~\ref{eureka} so that $(E_0)_{u,u}=a\geq \frac{1}{2}$, where $a=1-\frac{1}{k}$. If $m_j$ and $\ell_j$ are integers such that $\sum_{\lambda_j\neq 0}\ell_j\lambda_j=0$ and $m_j+\sum_{\lambda_j\neq 0}\ell_j=0$, then $m_j=0$. Invoking Lemma~\ref{eurekarem2} yields sharp $(1-\frac{2}{k})$-sedentariness at $u$. The last statement follows from the linear independence of the nonzero eigenvalues of $A(S_{k,k})$ when $4k+1$ is not a perfect square.
\end{proof}

It is natural to ask whether the internal vertices of $S_{k,\ell}$ are in general sedentary. We leave this as an open question.

\section{Other types of state transfer}\label{secOther}

By Proposition \ref{prop}(1c), a sedentary vertex cannot be involved in PGST. Hence, we ask, which types of state transfer can a sedentary vertex exhibit? Here, we show that there are sedentary families where each member graph exhibits proper fractional revival and local uniform mixing at a sedentary vertex.

Proper $(\alpha,\beta)$-\textit{fractional revival} (FR) occurs between $u$ and $v$ at time $t_1$ if $\alpha^2+\beta^2=1$, where $\alpha=\lvert U(t_1)_{u,u}\rvert$ and $\beta=\lvert U(t_1)_{u,v}\rvert\neq 0$. In \cite[Theorem 11]{Chan2020}, Chan et al.\ showed that proper Laplacian FR occurs between the apexes of $O_2\vee X$, where $X$ is a simple unweighted graph on $n$ vertices. Meanwhile, in \cite[Example 6.3]{Chan2019}, Chan et al.\ showed that $O_2\vee X$ exhibits proper Laplacian FR between its apexes, where $X$ is a simple unweighted $d$-regular graph on $n$ vertices. If $X$ is a simple positively weighted graph, it is shown in \cite{Monterde2023} that the apexes of $K_2\vee X$ do not admit proper Laplacian FR. Combining these facts with Corollaries \ref{corsedLLL}, \ref{corsedL} and \ref{corsedA}, we obtain families where each member graph exhibits (resp., does not exhibit) proper FR involving a sedentary vertex. This tells us that, unlike PGST, FR and sedentariness can occur together, although they do not always happen together.

\begin{example}
The following hold.
\begin{enumerate}
\item Each graph in the quasi-sedentary family of disconnected double cones in Corollary~\ref{corsedL} exhibits proper Laplacian FR between apexes. Moreover, each graph in the $C$-sedentary families of disconnected double cones in Corollary~\ref{corsedA} (1a-c) exhibits proper adjacency FR between apexes.
\item Each graph in the sedentary family of complete graphs on $n\geq 3$ vertices does not exhibit proper FR between any two vertices with respect to $A$ and $\mathcal{A}$. Moreover, each graph in the sedentary family of connected double cones in Corollary~\ref{corsedLLL} does not exhibit proper Laplacian FR between apexes.
\end{enumerate}
\end{example}

We say that $u$ admits \textit{(instantaneous) local uniform mixing} in $X$ at time $t_1$ if $\lvert U(t_1)_{u,v}\rvert=1/\sqrt{\lvert V(X)\rvert}$ for each vertex $v$ in $X$. We say that $X$ admits \textit{(instantaneous) uniform mixing} in $X$ at time $t_1$ if each vertex in $X$ admits local uniform mixing at time $t_1$.

\begin{proposition}
\label{prop1}
Let $0< C\leq 1$ and $\mathscr{F}$ be a $C$-sedentary family of graphs.
\begin{enumerate}
\item Almost all graphs in $\mathscr{F}$ do not exhibit local uniform mixing.
\item If the function $f$ in Definition \ref{defs} satisfies $f(\lvert V(X)\rvert)>\frac{1}{\sqrt{\lvert V(X)\rvert}}$ for all $X\in\mathscr{F}$, then each $X\in\mathscr{F}$ does not exhibit local uniform mixing at a sedentary vertex. 
\end{enumerate}
\end{proposition}

\begin{proof}
By assumption, for each $X\in \mathscr{F}$ and some vertex $u$ of $X$, we have $\lvert U_M(t)_{u,u}\rvert\geq f(\lvert V(X)\rvert)$ for all $t$, where $f(s)\rightarrow C>0$ as $s$ increases. Now, if $X\in \mathscr{F}$ admits local uniform mixing, then $\lvert U_M(t_1)_{u,u}\rvert=1/\sqrt{\lvert V(X)\rvert}$ for some time $t_1$. But since $C>0$ and $1/\sqrt{s}\rightarrow0$ as $s$ increases, only finitely many graphs in $\mathscr{F}$ can exhibit local uniform mixing. This proves (1), and (2) is straightforward.
\end{proof}

If $\mathscr{K}'$ is the family of complete graphs on $n\geq 5$ vertices, then from (\ref{K}), we may take $f$ such that $f(n)=1-\frac{2}{n}$. Since $f(n)>\frac{1}{\sqrt{n}}$ for all $n\geq 5$, no member of $\mathscr{K}'$ exhibits local uniform mixing by Proposition \ref{prop1}(2). Proposition \ref{prop1}(1), on the other hand, implies that only quasi-sedentary families exhibit local uniform mixing at a sedentary vertex as illustrated by our next examples.

\begin{example}
Let $\mathscr{F}$ be a family of cones on weighted $d$-regular graphs, where $0<d\leq 2$. Combining Proposition \ref{Godsil}(2) and \cite[Lemma 7.1]{SedQW}, we conclude that $\mathscr{F}$ is quasi-sedentary at the apex and each $X\in \mathscr{F}$ admits local uniform mixing at the apex with respect to $A$.
\end{example}

\begin{example}
Consider $Z_{k,3}$ in Example~\ref{stars}, which is a Cartesian power of $K_{1,3}$. This graph admits uniform mixing at $t_1=\frac{\pi}{3\sqrt{3}}$ \cite[Section 11]{Godsil2015UniformMO}, and so Example~\ref{stars}(2) implies that each graph in the quasi-sedentary family $\mathscr{F}=\{Z_{k,3}:k\geq 1\}$ admits adjacency uniform mixing.
\end{example}

Since Cartesian powers of $K_{3}$ admit uniform mixing at time $t_1=\frac{\pi}{9}$, the same result holds if we replace $Z_{k,3}$ in the previous example by $H(k,3)$. Moreover, this result applies to any $M$ because $H(k,3)$ is regular.

\section*{Acknowledgements}
I thank the University of Manitoba Faculty of Science and Faculty of Graduate Studies for the support. I thank Steve Kirkland, Sarah Plosker, Chris Godsil and Cristino Tamon for the helpful comments and useful discussions. I am also grateful to the referees for their suggestions that helped improve this paper.

\bibliographystyle{alpha}
\bibliography{mybibfile}
\end{document}